\documentclass{amsart}
\usepackage[latin1]{inputenc}
\usepackage{amsmath,comment}
\usepackage{amsfonts,bbm}
\usepackage{amssymb}
\usepackage{graphicx}
\usepackage{graphics}
\usepackage{lscape}

\usepackage{amsmath}
\numberwithin{equation}{section}

\newtheorem{theorem}{Theorem}
\numberwithin{theorem}{section}
\newtheorem{lemma}{Lemma}
\numberwithin{lemma}{section}
\newtheorem{prop}{Proposition}
\numberwithin{prop}{section}
\newtheorem{corol}{Corollary}
\newtheorem{problem}{Problem}
\numberwithin{corol}{section}
\newtheorem{remark}{Remark}
\numberwithin{remark}{section}

\def\P{\mathbb{P}}

\def\R{\mathbb{R}}
\def\C{\mathbb{C}}
\def\E{\mathbb{E}}
\def\ind{\mathbbm{1}}
  
\title[Analysis of the SQF algorithm]{Stationary analysis of the Shortest Queue First service policy}

\author{Fabrice Guillemin}
\address{Orange Labs, CORE/TPN, 2 Avenue Pierre Marzin, 22300 Lannion}
\email{fabrice.guillemin@orange.com}

\author{Alain Simonian}
\address{Orange Labs, CORE/TPN, 38-40 Rue du G\'en\'eral Leclerc, 92794 Issy-les-Moulineaux}
\email{alain.simonian@orange.com}

\begin{document}

\begin{abstract}
We analyze the so-called Shortest Queue First (SQF) queueing discipline whereby a unique server addresses queues in parallel by serving at any time that queue with the smallest workload. Considering a stationary system composed of two parallel queues and assuming Poisson arrivals and general service time distributions, we first establish the functional equations satisfied by the Laplace transforms of the workloads in each queue. We further specialize these equations to the so-called ``symmetric case'', with same arrival rates and identical exponential service time distributions at each queue; we then obtain a functional equation 
$$
M(z) = q(z) \cdot M \circ h(z) + L(z)
$$
for unknown function $M$, where given functions $q$, $L$ and $h$ are related to one  branch of a cubic polynomial equation. We study the analyticity domain of function $M$ and express it by a series expansion involving all iterates of function $h$. This allows us to determine empty queue probabilities along with the tail of the workload distribution in each queue. This tail appears to be  identical to that of the Head-of-Line preemptive priority system, which is the key feature desired for the SQF discipline.
\end{abstract}

\date{}

\maketitle


\section{Introduction}


Throughout this paper, we consider a unique server addressing two parallel queues numbered $\sharp$1 and $\sharp$2, respectively. Incoming jobs enter either queue and require  random service times; the server then processes jobs according to the so-called {Shortest Queue First} (SQF) policy. Specifically, let $U_1$ (resp. $U_2$) denote the workload in queue $\sharp$1 (resp. queue $\sharp$2) at a  given  time, including the remaining amount of work of the job possibly in service; the server then proceeds as follows:
\begin{itemize} 
\item Queue $\sharp$1 (resp. queue $\sharp$2) is served if $U_1 \neq 0$, $U_2 \neq 0$ and $U_1 \leq U_2$ (resp. if $U_1 \neq 0, U_2 \neq 0$ and $U_2 < U_1$);
\item If only one of the queues is empty, the non empty queue is served;
\item If both queues are empty, the server remains idle until the next job arrival.
\end{itemize}
In contrast to fixed priority disciplines where the server favors queues in some predefined order remaining unchanged in time (e.g., classical preemptive or non-preemptive head-of-line priority schemes), the SQF policy enables the server to dynamically serve  queues according to their current state.

The performance analysis of such a queueing discipline is motivated by the so-called SQF packet scheduling policy recently proposed to improve the quality of Internet access on high speed communication links. As discussed in \cite{Luca,Ost08}, SQF policy is designed to serve the shortest queue, i.e., the queue with the least number of waiting packets; in case of buffer overflow, packets are dropped from the longest queue. Thanks to this simple policy, the scheduler consequently prioritizes constant bit rate flows associated with delay-sensitive applications such as voice and audio/video streaming with intrinsic rate constraints; priority is thus implicitly given to smooth flows over data  traffic  associated with bulk  transfers that sense network bandwidth by filling buffers and sending packets in bursts.

In this paper, we consider the fluid version of the SQF discipline. Instead of packets (i.e., individual jobs), we deal with the workload (i.e., the amount of fluid in each queue). Since the fluid SQF policy  considers the shortest queue in volume, that is, in terms of workload, its performance is quantitatively described by the variations of variables $U_1$ and $U_2$. To simplify the analysis, we here suppose that the buffer capacity for both queues $\sharp$1 and $\sharp$2 is infinite. Moreover, we assume that incoming jobs enter either queue according to a Poisson process; in view of the above application context, one can argue that such Poisson arrivals can model  traffic where sources have peak rates significantly higher than that of the output link; such processes can hardly represent, however, the traffic variations of locally constant bit rate flows. This Poisson assumption, although limited in this respect, is nevertheless envisaged here in view of its mathematical tractability and as a first step towards the consideration of more complicated arrival patterns. 

The above framework enables us to define the pair $(U_1,U_2)$ representing the workloads in the stationary regime in each queue as a continuous-state Markov process in $\mathbb{R}^+ \times \mathbb{R}^+$. In the following, we determine the probability distribution of the couple $(U_1,U_2)$ by studying its Laplace transform. The problem can then essentially be formulated as follows.

\begin{problem}
\label{prob1}
Given the domain $\mathbf{\Omega} = \{(s_1,s_2) \in \mathbb{C}^2 \mid \Re(s_1) > 0, \Re(s_2) > 0\}$ and analytic functions $K_1$, $K_2$, $K$, $J_1$, and $J_2$ in $\mathbf{\Omega}$, determine two bivariate Laplace transforms $F_1$, $F_2$ and two univariate Laplace transforms $G_1$, $G_2$, analytic in $\mathbf{\Omega}$ and such that equations
$$
\left\{
\begin{array}{ll}
K_1(s_1,s_2) F_1(s_1,s_2) + K_2(s_1,s_2) G_2(s_2) = J_2(s_2) + H(s_1,s_2), 
\\
\\
K_1(s_1,s_2) G_1(s_1) + K_2(s_1,s_2) F_2(s_1,s_2) = J_1(s_1) - H(s_1,s_2),
\end{array} \right.
$$
for some analytic function $H$, together hold for in $ \mathbf{\Omega}$.
\end{problem}
\noindent
Note that each condition $K_1(s_1,s_2) = 0$ or $K_2(s_1,s_2) = 0$ with $(s_1,s_2) \in \mathbf{\Omega}$ brings the latter equations respectively to
$$
\left\{
\begin{array}{ll}
K_2(s_1,s_2)G_2(s_2) - H(s_1,s_2) = J_2(s_2) \; \; \; \mbox{and} \; \; \; K_1(s_1,s_2) = 0,
\\
\\
K_1(s_1,s_2)G_1(s_1) + H(s_1,s_2) = J_1(s_1) \; \; \; \mbox{and}  \; \; \; K_2(s_1,s_2) = 0.
\end{array} \right.
$$

To the best knowledge of the authors, the mathematical analysis of the SQF policy has not been addressed in the queueing literature. Some comparable queueing disciplines have nevertheless been studied:
\begin{itemize}
\item[-] The \textit{Longest Queue First} (LQF) symmetric policy is considered in \cite{Coh87}, where the author studies the stationary distribution of the number of waiting jobs $N_1$, $N_2$ in each queue; reducing the analysis to a boundary value problem on the unit circle, an integral formula  is provided for the generating function of the pair $(N_1,N_2)$;
\item[-] The \textit{Join the Shortest 2-server Queue} (JSQ), where an arriving customer joins the shortest queue if the number of waiting jobs in queues are unequal, is analyzed in \cite{Coh98}. The bivariate generating function for the number of waiting jobs is then determined as a meromorphic function in the whole complex plane, whose associated poles and residues are calculated recursively. 
\end{itemize}

While the above quoted studies address the stationary distribution of the number of jobs in each queue, we here consider the real-valued process $(U_1,U_2)$ of workload components whose stationary analysis requires the definition of its infinitesimal generator on the relevant functional space. Besides, the Laplace transform of the distribution of $(U_1,U_2)$ proves to be meromorphic not on the entire complex plane, but on the plane cut along some algebraic singularities (while the solution for JSQ exhibits polar singularities only); as a both quantitative and qualitative consequence, the decay rate of the stationary distribution at infinity for SQF may change according to the  system load from that defined by the smallest polar singularity to that defined by the smallest algebraic singularity. 

The organization of the paper is as follows. In Section \ref{MA}, a Markovian analysis provides the basic equations for the stationary distribution of the coupled queues; the functional equations verified by the relevant Laplace transforms are further derived in Section \ref{LTD}. In Section \ref{EXPI}, we specialize the discussion to the so-called symmetric exponential case where arrival rates are identical, and where service distribution are both exponential with identical mean; the functional equations are then specified and shown to involve a key cubic equation. Specifically, \textbf{Problem \ref{prob1}} for the symmetric case is shown to reduce to the following.

\begin{problem}
\label{prob2}
Solve the functional equation
$$
M(z) = q(z) \cdot M \circ h(z) + L(z),
$$
for function $M$, where given functions $q$, $L$ and $h$ are related to one branch of a key cubic polynomial equation $R(w,z) = 0$.
\end{problem}
\noindent
For real $z > 0$, the solution ${M}(z)$ is written in terms of a series involving all  iterates $h^{(k)}(z) = h \circ ... \circ h (z)$ for $k >0$. The analytic extension of solution $z \mapsto {M}(z)$ to some domain of the complex plane is further studied in Section \ref{SQFqueue}; this enables us to derive the empty queue probability along with the tail behavior of the workload distribution in each queue for the symmetric case. The latter is then compared to that of the associated preemptive Head of Line (HoL) policy. Concluding remarks are finally presented in Section~\ref{CL}. 

The proofs for basic functional equations as well as some technical results are deferred to the Appendix for better readability.


\section{Markovian analysis}
\label{MA}


As described in the Introduction, we assume that incoming jobs consecutively enter  queue $\sharp 1$ (resp. queue $\sharp 2$) according to a Poisson process with mean arrival rate $\lambda_1$ (resp. $\lambda_2$). Their respective service times are independent and identically distributed (i.i.d.) with probability distribution $\mathrm{d}B_1(x_1)$, $x_1 > 0$ (resp. $\mathrm{d}B_2(x_2)$, $x_2 > 0$) and mean  $1/\mu_1$ (resp. mean $1/\mu_2$). 

Let $\varrho_1 = \lambda_1/\mu_1$ (resp. $\varrho_2 = \lambda_2/\mu_2$) denote the mean load of queue $\sharp 1$ (resp. queue $\sharp 2$) and $\varrho = \varrho_1 + \varrho_2$ denote the total load of the system. Since the system is work conserving, its stability condition is $\varrho < 1$ and we assume it to hold in the rest of this paper. In this section, we first specify the evolution equations for the system and further derive its infinitesimal generator.


\subsection{Evolution equations}
\label{PE}


First consider the total workload $U = U_1 + U_2$ of the union of queues $\sharp 1$ and $\sharp 2$. For any work-conserving service discipline (such as SQF), the distribution of $U$ is independent of that discipline and equals that of the global single $M/G/1$ queue. The aggregate arrival process is Poisson with rate $\lambda = \lambda_1 + \lambda_2$ and the i.i.d. service times have the averaged distribution 
\begin{equation}
\mathrm{d}B(x) = \frac{\lambda_1}{\lambda} \mathrm{d}B_1(x) + \frac{\lambda_2}{\lambda} \mathrm{d}B_2(x), \; \; x > 0, 
\label{dB}
\end{equation}
with mean $\varrho/\lambda$. The stationary probability for the server to be in idle state, in particular, equals
\begin{equation}
\mathbb{P}(U = 0) = 1-\varrho.
\label{P0}
\end{equation}
\indent
Let $A_1(t)$ (resp. $A_2(t)$) be the number of job arrivals within time interval $[0,t[$ at queue $\sharp$1 (resp. queue $\sharp$2); if $\mathcal{T}_1^{(n)}$ (resp. $\mathcal{T}_2^{(n)}$) is the service time of the $n$-th job arriving at queue $\sharp$1 (resp. $\sharp$2), the total work  brought within $[0,t[$ into queue $\sharp$1 (resp. $\sharp 2$) equals $W_1(t) = \Sigma_{1 \leq n \leq A_1(t)} \; \mathcal{T}_1^{(n)}$ (resp. $W_2(t) = \Sigma_{1 \leq n \leq A_2(t)} \; \mathcal{T}_2^{(n)})$. Denoting by $U_1(t)$ (resp. $U_2(t)$) the workload in queue $\sharp 1$ (resp. $\sharp 2$) at time $t$, define indicator functions $I_1(t)$ and $I_2(t)$ by
\begin{equation}
\label{indicators}
\left\{
\begin{array}{ll}
I_1(t) = \mathbbm{1}_{ \{0 < U_1(t) \leq U_2(t)\} } + \mathbbm{1}_{ \{0 < U_1(t), U_2(t)=0\} },
\\ \\
I_2(t) = \mathbbm{1}_{\{0 < U_2(t) < U_1(t)\}} + \mathbbm{1}_{\{0 < U_2(t), U_1(t)=0\}},
\end{array} \right.
\end{equation}
respectively. 
With the above notation, the SQF policy governs workloads $U_1(t)$ and $U_2(t)$ according to the evolution equations
\begin{equation}\left\{
\begin{array}{ll}
\mathrm{d}U_1(t) = \mathrm{d}W_1(t) - I_1(t) \; \mathrm{d}t, 
\\ \\
\mathrm{d}U_2(t) = \mathrm{d}W_2(t) - I_2(t) \; \mathrm{d}t,
\label{path}
\end{array} \right.
\end{equation}                  
for $t > 0$ and some initial conditions $U_1(0) \geq 0$, $U_2(0) \geq 0$. This defines the pair $\mathbf{U}(t) = (U_1(t),U_2(t))$, $t \geq 0$, as a Markov process with state space $\mathcal{U} = \mathbb{R}^+ \times \mathbb{R}^+$ (see Figure~\ref{Fig1} for sample paths of process $\mathbf{U}(t) = (U_1(t),U_2(t))$). 

\begin{figure}[hbtp]
\scalebox{.8}{\includegraphics[width=15cm, trim=3cm 17cm 0 2cm,clip]{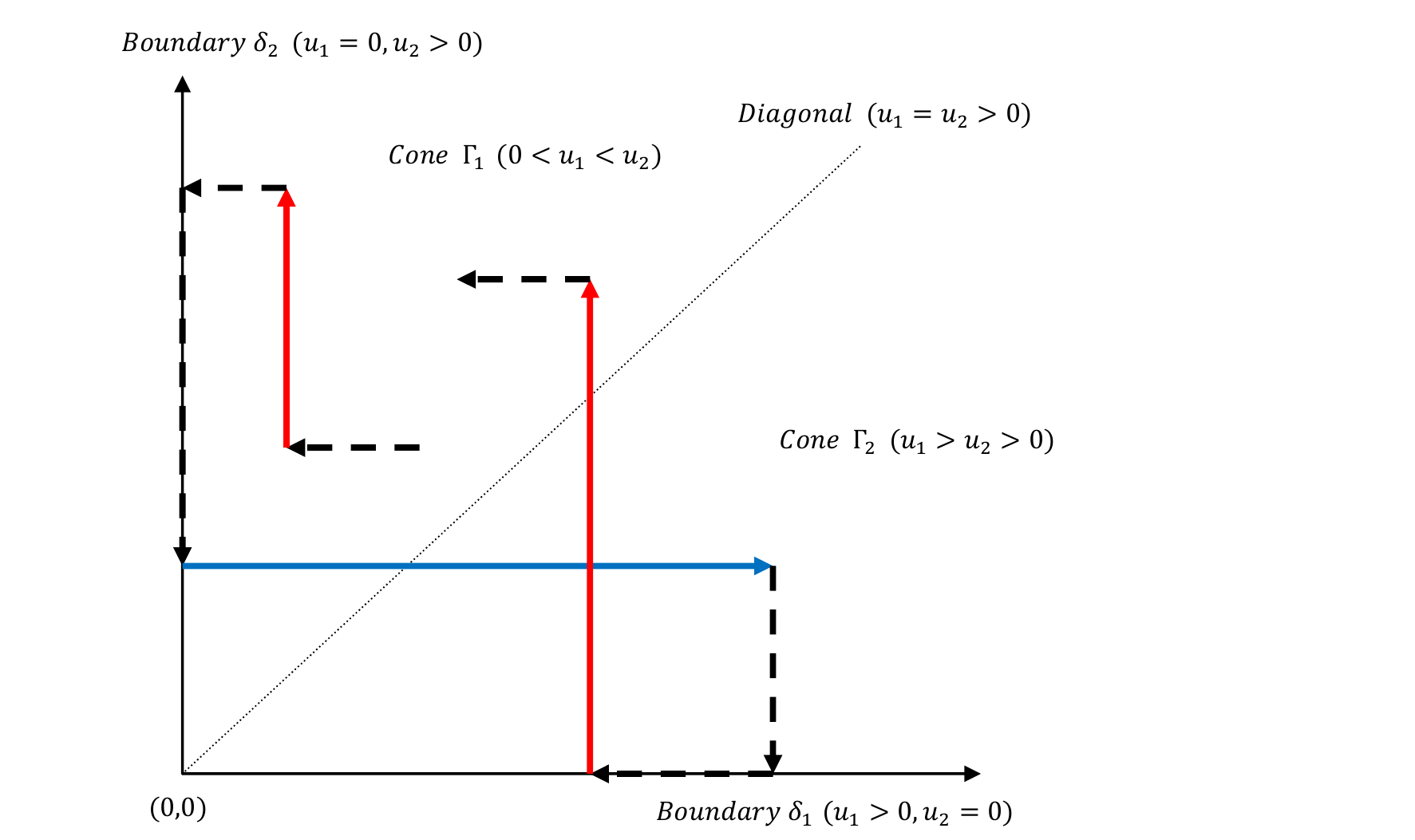}}
\caption{{Sample path of process $\mathbf{U}(t) = (U_1(t),U_2(t))$ (dashed lines) with job arrivals at queue $\sharp 1$ (solid blue line) or queue $\sharp 2$ (solid red line).}
\label{Fig1}}
\end{figure}

As a first result, integrating each equation~(\ref{path}) over interval $[0,t]$, dividing each side by $t$ and letting $t \uparrow +\infty$ implies $\lim_{t \uparrow +\infty} U_1(t)/ t = 0$ and 
$$
\lim_{t \uparrow +\infty} \frac{1}{t}\int_0^t I_1(s) \mathrm{d}s = \mathbb{P}(I_1 = 1), \; \; \lim_{t \uparrow +\infty} \frac{W_1(t)}{t} = \varrho_1
$$
almost surely (along with similar limits for integrals related to $U_2$, $I_2$ and $W_2$), and equating these limits readily provides identites
\begin{equation}
\mathbb{P}(I_1 = 1) = \varrho_1 \quad \mbox{and} \quad  \mathbb{P}(I_2 = 1) = \varrho_2
\label{Y1Y2}
\end{equation}
for the stationary probability that the server treats queue $\sharp$1 and $\sharp$2, respectively; equivalently, the latter identities read 
\begin{align}
\mathbb{P}(0 \leq U_1 \leq U_2) - \P(U_1=0) + \P(U_2=0) & = 1-\varrho_2,
\nonumber \\
\mathbb{P}(0 \leq U_2 \leq U_1) - \P(U_2=0) + \P(U_1=0) & = 1-\varrho_1.
\nonumber
\end{align}
In the symmetric case when arrival rates are equal and service times have identical distribution, i.e., $\lambda_1 = \lambda_2 = \lambda/2$ and $\mu_1 = \mu_2$, the above relations give 
$$\mathbb{P}(0\leq U_1\leq U_2) = \mathbb{P}(0\leq U_2\leq U_1) = 1 - \varrho/2.
$$

\begin{remark}
The discrepancy in inequalities $0 < U_1 \leq U_2$ and $0 < U_2 < U_1$ defining the service policy (when both queues are non empty) does not favor queue $\sharp 1$ with respect to queue $\sharp 2$, since event $U_1 = U_2 \neq 0$ has probability 0; in fact, assuming for instance $0 < U_2(t) < U_1(t)$ at some time $t$, we have 
$$
U_1(t+0) = U_2(t+0)
$$
if a job arrival occurs with service time of amount exactly $U_1 - U_2$, which has probability 0 for any service time distribution. Hence, the distribution of process $\mathbf{U}$ does not give a positive probability to the diagonal $\{(u_1,u_1) \in \mathcal{U} \mid u_1 > 0\}$ in state space $\mathcal{U} = \R^+\times \R^+$.  
\end{remark}


\subsection{Infinitesimal generator}
\label{IG}


We now address the determination of the stationary distribution function  $\Phi(u_1,u_2)=\P(U_1\leq u_1,U_2 \leq u_2)$, $u_1 \geq 0, u_2 \geq 0$, of the bivariate workload process $\mathbf{U}$. In order to define the class of stationary distribution   $\Phi$, we further assume that
\begin{itemize}
\item [\textbf{A.1}] Distribution $ \Phi$ has a regular density $\varphi_1(u_1,u_2)$ (resp. $\varphi_2(u_1,u_2)$) at any point $(u_1,u_2)$ such that $0 < u_1 < u_2$ (resp. $0 < u_2 < u_1$);
\item [\textbf{A.2}] Distribution $\Phi$ has a regular density $\psi_1(u_1)$ (resp. $\psi_2(u_2)$) at any point $u_1 > 0$ (resp. $u_2 > 0$) on the boundary $\{(u_1,u_2) \mid u_1> 0, u_2=0\}$ (resp. on the boundary $\{(u_1,u_2) \mid u_1=0, u_2>0\}$).
\end{itemize}
(A real-valued function is here said to be regular if it is continuous and bounded over its definition domain.) In the rest of this paper, assumptions \textbf{A.1}-\textbf{A.2} for the existence of regular densities will be confirmed by exhibiting their Laplace transforms; the uniqueness of the stationary distribution then \textit{a posteriori} justifies such assumptions. An \textit{a priori} justification for the existence of densities would otherwise imply the use of Malliavin Calculus \cite{Hir92, Rev10} on the Poisson space.

Using \eqref{P0}, we have $\mathbb{P}(U_1 = 0, U_2 = 0) = \mathbb{P}(U_1 + U_2=0) = 1 - \varrho$ in the stationary regime; following assumptions \textbf{A.1}-\textbf{A.2} above, we can then write
\begin{multline}
\label{dPhi}
\mathrm{d}\Phi(u_1,u_2) =  \psi_1(u_1)\ind_{\{u_1 > 0\}}\mathrm{d}u_1 \otimes \delta_0(u_2) + \psi_2(u_2)\ind_{\{u_2 > 0\}}\mathrm{d}u_2 \otimes \delta_0(u_1) \\ + (\varphi_1(u_1,u_2)\ind_{\{0 < u_1 < u_2\}} + \varphi_2(u_1,u_2)\ind_{\{0 < u_2 < u_1\}})\mathrm{d}u_1\mathrm{d}u_2  +  (1-\varrho)\delta_{0,0}(u_1,u_2) 
\end{multline}
for all $(u_1,u_2) \in \mathcal{U}$, where $\delta_{0,0}$ (resp. $\delta_0$) is the Dirac distribution at point $(0,0)$ (resp. at point $0$).
\\
\indent
Let us now characterize the stationary distribution of process $\mathbf{U} = (U_1,U_2)$ by means of its infinitesimal generator $\mathcal{A}$ defined by 
$$
\forall \; \mathbf{u} \in \mathcal{U}, \; \; \; \mathcal{A}\theta(\mathbf{u}) = \lim_{h \rightarrow 0} \frac{1}{h} \Big [\mathbb{E}\theta(\mathbf{U}_{t+h} \mid \mathbf{U}_t = \mathbf{u}) - \theta(\mathbf{u}) \Big ]
$$
where the limit is uniform with respect to $\mathbf{u} \in \mathcal{U}$ (see \cite[p.~175]{Wen81} or \cite[p.~8, p.~377]{Eth05}); the symbol $\theta$ denotes any function for which the latter limit exists. In the following, we denote by  $\mathcal{C}^2_b(\mathcal{U})$ the set of functions $\theta: \mathcal{U} \rightarrow \mathbb{C}$ everywhere bounded, twice differentiable with bounded first and second derivatives in $\mathcal{U}$. Further, introduce positive cones
\begin{equation}
\Gamma_1 = \{(u_1, u_2) \in \mathcal{U} \; \mid \; 0 < u_1 < u_2\}, \; \; \; \Gamma_2 = \{(u_1, u_2) \in \mathcal{U} \; \mid \; 0 < u_2 < u_1\},
\label{cones12}
\end{equation}
along with boundaries (see Figure~\ref{Fig1}) 
\begin{equation}
\delta_1 = \{(u_1,0) \in \mathcal{U} \; \mid \; u_1 > 0\}, \; \; \; \delta_2 = \{(0,u_2) \in \mathcal{U} \; \mid \; u_2 > 0\}.
\label{axes12}
\end{equation}
We can then state the following.

\begin{prop} 
With the notation (\ref{cones12})-(\ref{axes12}), the infinitesimal generator $\mathcal{A}$ of process $\mathbf{U} = (U_1,U_2)$ is given by
\begin{align}
\mathcal{A}\theta(\mathbf{u}) = & - \frac{\partial \theta}{\partial u_1}(\mathbf{u})\ind_{\{\mathbf{u} \in  \Gamma_1 \cup \delta_1\}} -  \frac{\partial \theta}{\partial u_2 }(\mathbf{u})\ind_{\{\mathbf{u} \in \Gamma_2 \cup \delta_2\}} 
\nonumber \\ 
& + \lambda_1 \mathbb{E}[\theta(\mathbf{u}+\mathcal{T}_1\mathbf{e}_1)-\theta(\mathbf{u})]+\lambda_2 \mathbb{E}[\theta(\mathbf{u}+\mathcal{T}_2\mathbf{e}_2)-\theta(\mathbf{u})]
\label{gen1}
\end{align}
for all $\mathbf{u} \in \mathcal{U}$ and any test function $\theta \in \mathcal{C}^2_b(\mathcal{U})$, where $\mathcal{T}_1$ (resp. $\mathcal{T}_2$) denotes the generic service time of jobs arriving at queue $\sharp$1 (resp. queue $\sharp$2) and with vectors $\mathbf{e}_1 = (1,0)$, $\mathbf{e}_2 = (0,1)$.
\label{generator}
\end{prop}

\begin{proof}
Using evolution equations (\ref{path}) and given $\mathbf{U}_t = \mathbf{u}$, expression (\ref{gen1}) is easily derived from uniform estimates (with respect to $\mathbf{u} \in \mathcal{U}$) for the distribution of the number of jumps of process $(\mathbf{U}_t)$ on any interval $[t,t+h]$ (all intervening Poisson processes have rates lower  than $\lambda = \lambda_1 + \lambda_2$) and for drift rates (when non zero, the service rate is the constant $-1$) .
\end{proof}

Once generator $\mathcal{A}$ is determined, the stationary distribution $\Phi$ of $(\mathbf{U}_t)_{t \geq 0}$ is known (see \cite[p.~189]{Wen81} or \cite[p.~239]{Eth05}) to satisfy
\begin{equation}
\forall \; \theta \in \mathcal{C}^2_b(\mathcal{U}), \; \; \; \int_{\mathcal{U}} \mathcal{A}\theta(\mathbf{u}) \mathrm{d}\Phi(\mathbf{u}) = 0.
\label{stationary}
\end{equation}


\section{Laplace transforms derivation}
\label{LTD}


Following the prerequisites of Section~\ref{MA}, we now study integral equation~\eqref{stationary}. Since the problem  is linear in unknown distribution $\Phi$, it is tractable  through Laplace transform techniques. 


\subsection{Functional equations}
\label{FE}


Let $\mathbf{\Omega} = \{(s_1,s_2) \in \mathbb{C}^2 \mid \Re(s_1) > 0, \Re(s_2) > 0\}$ and its closure $\overline{\mathbf{\Omega}}$. Assumptions $\textbf{A.1-A.2}$ in Section~\ref{IG}, for the existence of regular densities $\varphi_1$ and $\psi_1$ with respective support $\Gamma_1$ and $\delta_1$ (see Equations~(\ref{cones12})-(\ref{axes12})) enable us to define their Laplace transforms $F_1$, $G_1$ by
\begin{equation}
F_1(s_1,s_2) = \int_{\Gamma_1} e^{-\mathbf{s} \cdot \mathbf{u}} \varphi_1(\mathbf{u})\mathrm{d}\mathbf{u}, \; \; \; \; G_1(s_1) = \int_{\delta_1} e^{-s_1u_1}\psi_1(u_1)\mathrm{d}u_1
\label{FF}
\end{equation}
for $\mathbf{s} = (s_1, s_2) \in \overline{\mathbf{\Omega}}$, where $\mathbf{s} \cdot \mathbf{u} = s_1u_1 + s_2u_2$; using the expectation operator, definitions (\ref{FF}) equivalently read
$$
F_1(s_1,s_2) = \E\big[e^{-s_1U_1-s_2U_2}\ind_{\{0< U_1 < U_2\}}\big], \; \; G_1(s_1) = \E\big[e^{-s_1U_1}\ind_{\{0=U_2<U_1\}}\big].
$$
The Laplace transforms $F_2$ and $G_2$ of regular densities $\varphi_2$ and $\psi_2$ with respective support $\Gamma_2$ and $\delta_2$ (see Equations~(\ref{cones12})-(\ref{axes12})) are similarly defined by
\begin{equation}
F_2(s_1,s_2) = \int_{\Gamma_2} e^{-\mathbf{s} \cdot \mathbf{u}} \varphi_2(\mathbf{u})\mathrm{d}\mathbf{u}, \; \; \; \; G_2(s_2) = \int_{\delta_2} e^{-s_2u_2}\psi_2(u_2)\mathrm{d}u_2
\label{FFb}
\end{equation}
for $\mathbf{s} = (s_1, s_2) \in \overline{\mathbf{\Omega}}$; equivalent definitions can be similarly written in terms of the expectation operator. Expression (\ref{dPhi}) for distribution $\mathrm{d}\Phi$ and the above definitions then enable to define the Laplace transform $F$ of the pair $(U_1,U_2)$ by
\begin{equation}
F(s_1,s_2) = 1 - \varrho + F_1(s_1,s_2) + G_1(s_1) + F_2(s_1,s_2) + G_2(s_2)
\label{Ldef}
\end{equation}
for $(s_1,s_2) \in \overline{\mathbf{\Omega}}$.

Finally, let $b_1(s_1) = \mathbb{E}(e^{-s_1\mathcal{T}_1})$ (resp. $b_2(s_2) = \mathbb{E}(e^{-s_2\mathcal{T}_2})$) denote the Laplace transform of service time $\mathcal{T}_1$ (resp. $\mathcal{T}_2$) at queue $\sharp 1$ (resp. queue $\sharp 2$) 
for $\Re(s_1) \geq 0$ (resp. $\Re(s_2) \geq 0$); set in addition
\begin{equation}
\left\{
\begin{array}{ll}
K(s_1,s_2)=\lambda-\lambda_1 b_1(s_1)-\lambda_2 b_2(s_2), 
\\ \\
K_1(s_1,s_2) = s_1 - K(s_1,s_2), \; \; \; K_2(s_1,s_2) = s_2 - K(s_1,s_2),
\end{array} \right.
\label{Ker}
\end{equation}
and
\begin{equation}
\left\{
\begin{array}{ll}
J_1(s_1) = (1-\varrho)(\lambda - \lambda_1b_1(s_1)) - \psi_2(0),
\\ \\
J_2(s_2) = (1-\varrho)(\lambda - \lambda_2b_2(s_2)) - \psi_1(0).
\end{array} \right.
\label{J1J2}
\end{equation}

\begin{prop} \textbf{a)} Transforms $F_1$, $G_1$ and $F_2$, $G_2$ together satisfy
\begin{equation}
K_1(s_1,s_2)H_1(s_1,s_2) + K_2(s_1,s_2)H_2(s_1,s_2) = (1-\varrho)K(s_1,s_2)
\label{Fonct}
\end{equation}
for $(s_1, s_2) \in {\mathbf{\Omega}}$, where $H_1 = F_1 + G_1$ and $H_2 = F_2 + G_2$.

\textbf{b)} Transforms $F_1$ and $G_2$ (resp. $F_2$, $G_1$) satisfy
\begin{equation}
\left\{ 
\begin{array}{lll}
K_1(s_1,s_2)F_1(s_1,s_2) + K_2(s_1,s_2) G_2(s_2) = J_2(s_2) + H(s_1,s_2),
\\
\\
K_2(s_1,s_2)F_2(s_1,s_2) + K_1(s_1,s_2) G_1(s_1) = J_1(s_1) - H(s_1,s_2)
\end{array} \right.
\label{Fonctcomp}
\end{equation}
for $(s_1, s_2) \in {\mathbf{\Omega}}$, with 
\begin{multline}
H(s_1,s_2) = \lambda_1 \mathbb{E} \left [ e^{-s_1U_1-s_2U_2}\ind_{\{0 \leq U_1 < U_2\}}e^{-s_1\mathcal{T}_1}\ind_{\{\mathcal{T}_1 > U_2-U_1\}} \right ] 
 \\ - \lambda_2\mathbb{E} \left [ e^{-s_1U_1-s_2U_2}\ind_{\{0 \leq U_2 < U_1\}}e^{-s_2\mathcal{T}_2}\ind_{\{\mathcal{T}_2 > U_1-U_2\}} \right ].
\label{H}
\end{multline}

\textbf{c)} Constants $\psi_1(0)$ and $\psi_2(0)$ satisfy relation $\psi_1(0) + \psi_2(0) = \lambda(1-\varrho)$.
\label{resol}
\end{prop}

\begin{proof}
\textbf{a)} Fix $(s_1, s_2) \in \mathbf{\Omega}$. The test function $\theta(\mathbf{u}) = e^{-\mathbf{s} \cdot \mathbf{u}}$, $\mathbf{u} = (u_1,u_2) \in \mathcal{U}$, belongs to $\mathcal{C}^2_b(\mathcal{U})$ and has derivatives $\partial \theta/\partial u_1 = -s_1\theta$, $\partial \theta/\partial u_2 = -s_2\theta$. Besides, we have $\theta(\mathbf{u}+\mathcal{T}_1\mathbf{e}_1) = e^{-s_1T_1}\theta(\mathbf{u})$ hence $\mathbb{E}[\theta(\mathbf{u}+\mathcal{T}_1\mathbf{e}_1) - \theta(\mathbf{u})] = (b_1(s_1)-1)\theta(\mathbf{u})$, and similarly $\mathbb{E}[\theta(\mathbf{u}+\mathcal{T}_2\mathbf{e}_2) - \theta(\mathbf{u})] = (b_2(s_2)-1)\theta(\mathbf{u})$. Applying Proposition \ref{generator}, formula (\ref{gen1}) for $\mathcal{A}\theta(\mathbf{u})$ then yields
$$
\mathcal{A}\theta(\mathbf{u}) = (s_1\ind_{\mathbf{u} \in \Gamma_1 \cup \delta_1} + s_2\ind_{\mathbf{u} \in \Gamma_2 \cup \delta_2})\theta(\mathbf{u}) 
- K(s_1,s_2)\theta(\mathbf{u})
$$ 
with $K(s_1,s_2)$ defined in (\ref{Ker}). Integrating that expression of $\mathcal{A}\theta(\mathbf{u})$ over closed quarter plane $\mathcal{U}$ with respect to distribution $\mathrm{d}\Phi$ and using Assumptions \textbf{A.1}-\textbf{A.2}, Relation~\eqref{stationary} then gives  
$$
\int_{\mathcal{U}} \mathcal{A}\theta(\mathbf{u}) \mathrm{d}\Phi(\mathbf{u}) = s_1H_1(s_1,s_2) + s_2H_2(s_1,s_2) - F(s_1,s_2)K(s_1,s_2) = 0
$$ 
with $H_1 = F_1 + G_1$ and $H_2 = F_2 + G_2$; using (\ref{Ldef}) finally provides (\ref{Fonct}).
\smallskip

\textbf{b)} As detailed in Appendix \ref{A2}, there exists a family of functions $\chi_\varepsilon: \mathcal{U} \rightarrow \mathbb{R}$ with $\varepsilon > 0$, such that $\forall \; \mathbf{u} \in \mathcal{U}$, $\lim_{\varepsilon \downarrow 0} \chi_\varepsilon(\mathbf{u}) = \ind_{\{\mathbf{u} \in \Gamma_1\}}$ \textit{and} $\chi_\varepsilon \in \mathcal{C}_b^2(\mathcal{U})$. For given $\Re(s_1) > 0$, $\Re(s_2) > 0$, the function $\theta_\varepsilon$ defined by $\theta_\varepsilon(\mathbf{u}) = e^{-\mathbf{s} \cdot \mathbf{u}}\chi_\varepsilon(\mathbf{u})$, $\mathbf{u} \in \mathcal{U}$, therefore belongs to $\mathcal{C}^2_b(\mathcal{U})$ and satisfies $\lim_{\varepsilon \downarrow 0} \theta_\varepsilon = \theta$ pointwise in $\mathcal{U}$ with
$$
\theta(\mathbf{u}) = e^{-\mathbf{s} \cdot \mathbf{u}} \cdot \ind_{\{\mathbf{u} \in \Gamma_1\}}, \; \; \mathbf{u} \in \mathcal{U}
$$
(note that $\theta \notin \mathcal{C}^2_b(\mathcal{U})$). Apply then formula (\ref{gen1}) to regularized test function $\theta_\varepsilon$ and integrate this expression over $\mathcal{U}$ against distribution $\mathrm{d}\Phi$ to define 
\begin{equation}
\mathcal{M}(\varepsilon) = \int_{\mathcal{U}} \mathcal{A}\theta_\varepsilon(\mathbf{u}) \mathrm{d}\Phi(\mathbf{u}).
\label{mepsilon}
\end{equation}
In view of (\ref{stationary}), we have $\mathcal{M}(\varepsilon) = 0$ and, provided that $\mathcal{M}(\varepsilon)$ has a finite limit $\mathcal{M}$ as $\varepsilon \downarrow 0$, we must have $\mathcal{M} = 0$. The detailed calculation of that limit $\mathcal{M}$ (depending on the pair $(s_1,s_2)$) is performed in Appendix \ref{A2} and condition $\mathcal{M} = 0$ is shown to reduce to first  equation (\ref{Fonctcomp}). Exchanging indices 1 and 2 provides second  equation (\ref{Fonctcomp}), after noting that $H(s_1,s_2)$ changes into $- H(s_1,s_2)$.
\smallskip

\textbf{c)} Adding equations (\ref{Fonctcomp}) gives (\ref{Fonct}) if and only if $\psi_1(0) + \psi_2(0) = \lambda(1-\varrho)$ holds.
\end{proof}

\begin{remark}
Computing $F(s,s) = 1 - \varrho + H_1(s,s) + H_2(s,s)$ by letting  $s_1 = s_2 = s$ in (\ref{Fonct}) readily gives
\begin{equation}
F(s,s) = \frac{s(1-\varrho)}{s - K(s,s)}, \; \; \; \Re(s) > 0,
\label{PK}
\end{equation}
with $K(s,s) = \lambda - \lambda_1b_1(s) - \lambda_2b_2(s)$. Identity (\ref{PK}) is obviously Pollaczek-Khintchin formula \cite[p.~48, p.~339] {Rob03}  for the transform $F(s,s) = \mathbb{E}(e^{-sU})$ of the total workload $U = U_1 + U_2$ in the global $M/G/1$ queue, with i.i.d. service times having distribution $\mathrm{d}B$ defined by (\ref{dB}). 
\end{remark}

\begin{corol} 
\label{coroltech}
Let $H$ be defined by (\ref{H}). Transform $G_1$ satisfies
\begin{equation}
(s_1 - s_2)G_1(s_1) = J_1(s_1) - H(s_1,s_2)
\label{FonctG1}
\end{equation}
for $(s_1, s_2) \in \mathbf{\Omega}$ such that $K_2(s_1,s_2) = 0$. Similarly, transform $G_2$ satisfies
\begin{equation}
(s_2 - s_1)G_2(s_2) = J_2(s_2) + H(s_1,s_2)
\label{FonctG2}
\end{equation}
for $(s_1, s_2) \in \mathbf{\Omega}$ such that $K_1(s_1,s_2) = 0$.
\label{C.2}
\end{corol}

\begin{proof}
Function $F_2(s_1,s_2)$ is finite for any given $(s_1, s_2) \in \mathbf{\Omega}$; if  $K_2(s_1,s_2) = 0$, the product $K_2(s_1,s_2)F_2(s_1,s_2)$ is therefore zero. As
$$
K_2(s_1,s_2) = 0 \Rightarrow K_1(s_1,s_2) = s_1 - s_2,
$$
second equation (\ref{Fonctcomp}) then implies (\ref{FonctG1}). Relation (\ref{FonctG2}) is similarly derived.
\end{proof}


\subsection{Analytic continuation}
\label{AC}


In this section, we first compare the SQF system with the HoL queue, where one queue has Head of Line (HoL) priority over the other; such a comparison then enables us to extend the analyticity domain of Laplace transforms $F_1$, $F_2$ and $G_1$, $G_2$. 
\\
\indent
Let $\overline{U}_j(t)$, $j \in \{1,2\}$, denote the workload in queue $\sharp j$ when the other queue has HoL priority; similarly, let $ \underline{U}_j(t)$ denote the workload in queue $\sharp j$ when this queue has HoL priority over the other. Finally, given two real random variables $X$ and $Y$, $Y$ is said to  dominate $X$ in the strong  order sense (for short,  $X \leq_{st} Y$) if and only if $\E f(X)\leq \E f(Y)$ for any positive non-decreasing measurable function $f$.
\begin{prop}
Workload $U_j(t)$ verifies 
\begin{equation}
\underline{U}_j(t) \leq_{st} U_j(t) \leq_{st} \overline{U}_j(t)
\label{stochlowup}
\end{equation}
for all $t \geq 0$.
\label{domstoch}
\end{prop}
\begin{proof}
We clearly have $I_1(t) \geq \ind_{\{0 < U_1(t), U_2(t)=0\}}$ almost surely for all $t \geq 0$, where $I_1(t)$ is defined by \eqref{indicators}. Equation~(\ref{path}) consequently entails that $U_1(t) \leq \overline{U}_1(t)$ pathwise, which implies the strong stochastic domination. Similarly, we have $I_1(t) \leq \ind_{\{0 < U_1(t)\}}$ almost surely for all $t \geq 0$ and (\ref{path}) entails $U_1(t) \geq \underline{U}_1(t)$ pathwise, hence the strong stochastic domination. 
\end{proof}

Assume that random variable $\overline{U}_j =\lim_{t \uparrow +\infty}\overline{U}_j(t)$ has an analytic Laplace transform $s \mapsto \E(e^{-s\overline{U}_j})$ in the domain $\{s \in \C \mid \; \Re(s )> \widetilde{s}_j\}$ for some real $\widetilde{s}_j<0$.

\begin{corol}
Laplace transform $F_1$ can be analytically extended to domain
$$
\widetilde{\mathbf{\Omega}}_1 = \{(s_1,s_2) \in \C^2 \mid \; \Re(s_2) > \max(\widetilde{s}_2,\widetilde{s}_2-\Re(s_1))\},
$$
and transform $G_2$ can be analytically extended to $\widetilde{\omega}_2 = \{s_2 \in \C \mid \; \Re(s_2)>\widetilde{s}_2\}$. 

Similarly, transform $F_2$ can be analytically extended to
$$
\widetilde{\mathbf{\Omega}}_2 = \{(s_1,s_2) \in \C^2 \mid \; \Re(s_1) > \max(\widetilde{s}_1,\widetilde{s}_1-\Re(s_2)\},
$$
and $G_1$ can be analytically extended to $\widetilde{\omega}_1 = \{s_1\in \C \mid \; \Re(s_1)>\widetilde{s}_1\}$.
\label{extensions}
\end{corol}

\begin{proof}
Assume first that $s_1$ and $s_2$ are real with $s_1<0$; given $U_1< U_2$, we have $-s_1U_1-s_2U_2 <  -(s_1+s_2) {U}_2$; using the domination property $U_2 \leq _{st} \overline{U}_2$ of Proposition \ref{domstoch} and the previous inequality, definition (\ref{FF}) of $F_1$ on $\overline{\mathbf{\Omega}}$ entails
$$
F_1(s_1,s_2) = \E\big[e^{-s_1U_1-s_2U_2}\ind_{\{0< U_1 < U_2\}}\big] \leq \mathbb{E}\big[e^{-(s_1+s_2)\overline{U}_2}\big];
$$
we then deduce that $F_1$ can be analytically continued to any point $(s_1,s_2)$ verifying $\Re(s_1)<0$ and $\Re(s_1+s_2) > \widetilde{s}_2$. Assuming now that $s_1\geq 0$ and $s_2< 0$, domination property $U_2 \leq_{st} \overline{U}_2$ yields  $-s_1U_1-s_2U_2 \leq -s_2U_2 \leq_{st} -s_2 \overline{U}_2$ and definition (\ref{FF}) of $F_1$ on $\overline{\mathbf{\Omega}}$ entails in turn
$$
F_1(s_1,s_2) = \E\big[e^{-s_1U_1-s_2U_2}\ind_{\{0< U_1 < U_2\}}\big] \leq \mathbb{E}\big[e^{- s_2 \overline{U}_2}\big];
$$
$F_1$ can therefore be analytically continued to any point $(s_1,s_2)$ verifying $\Re(s_1)>0$ and $\Re(s_2)> \widetilde{s}_2$. We conclude that $F_1$ can be analytically continued to domain $\widetilde{\mathbf{\Omega}}_1$, as claimed.

Writing definition (\ref{FFb}) of $G_2$ as $G_2(s_2) = \E\big[e^{-s_2U_2}\ind_{\{0=U_1<U_2\}}\big]$ for $\Re(s_2) \geq 0$, the same type of arguments as above enables us to analytically continue function $G_2$ to any point verifying $\Re(s_2)>\widetilde{s}_2$.
\end{proof}

Domains $\widetilde{\mathbf{\Omega}}_1$ and $\widetilde{\mathbf{\Omega}}_2$ are illustrated in Figure~\ref{Fig3} (assuming $\widetilde{s}_2 < \widetilde{s}_1$ for instance). 

\begin{figure}[ht]
\centering
\scalebox{1}{\includegraphics[width=
10cm, trim = 2cm 12cm 3cm 4.3cm, clip]{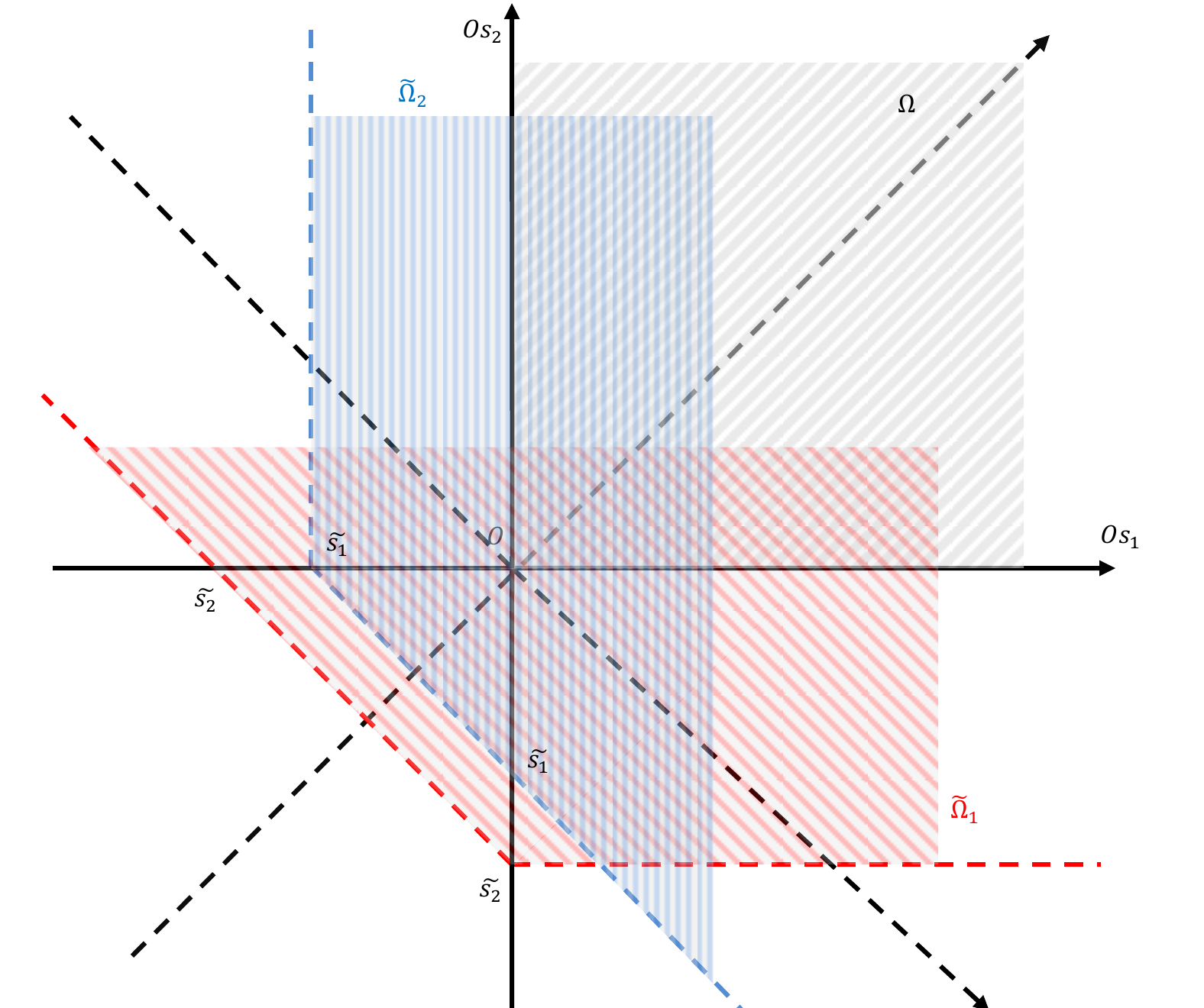}}
\caption{{Extension domains $\widetilde{\mathbf{\Omega}}_1$ (dotted red) and  $\widetilde{\mathbf{\Omega}}_2$ (dotted blue) in $\mathbb{R}^2$.}}
\label{Fig3}
\end{figure}

Following Proposition \ref{resol} and Corollary~\ref{coroltech}, the determination of Laplace transforms $F_1$, $F_2$, $G_1$ and $G_2$ critically depends on both the determination of auxiliary bivariate function $H$ generally defined in (\ref{H}) and  the solutions to equations $K_1(s_1,s_2)=0$ and $K_2(s_1,s_2)=0$. The latter, however, may be very intricate to compute for general service time distributions. 

To make the resolution more tractable, we will now introduce some specific assumptions. First, service times are assumed to be exponentially distributed; this readily provides a more explicit expression for function $H$.
\begin{prop}
In the case of exponentially distributed service times, we have
\begin{equation}
H(s_1,s_2) =\frac{\lambda_1\mu_1}{\mu_1+s_1} M_1\left(\frac{s_1+s_2}{2}\right) - \frac{\lambda_2\mu_2}{\mu_2+s_2} M_2\left(\frac{s_1+s_2}{2}\right),
\label{Hbis}
\end{equation}
where
$$
\left\{
\begin{array}{ll}
M_1(z) = G_2(2z+\mu_1)+F_1(-\mu_1,2z+\mu_1),
\\ \\
M_2(z) = G_1(2z+\mu_2)+F_2(2z+\mu_2,-\mu_2)
\end{array} 
\right.
$$
are analytically defined for $\Re(z) > \max(\widetilde{s}_1,\widetilde{s}_2)/2$.
\label{Hexp}
\end{prop}
\noindent
The proof of Proposition \ref{Hexp} is deferred to Appendix \ref{A3}. Expression (\ref{Hbis}) consequently reduces the determination of function $H$ to that of two univariate functions $M_1$ and $M_2$. 

In the rest of this paper, we further assume that the Poisson arrival rates and service time distributions in each queue are equal, the so-called ``symmetric (exponential) case''. Because of its technical complexity, the asymmetric case will be treated in a forthcoming paper \cite{assymSQF}.


\section{Analytic properties for the symmetric case}
\label{EXPI}


As previously motivated, we assume from now on that 
\begin{itemize}
\item Poisson arrival rates are equal, namely $\lambda_1 = \lambda_2 = \lambda/2$;
\item service times in both queues are exponentially distributed  with identical parameter $\mu$, i.e., $\mu_1=\mu_2=\mu$;
\end{itemize}
the Laplace transform of the service time distribution is then $b(s) = \mu/(s + \mu)$. 

By the latter symmetry assumption, queues $\sharp 1$ and $\sharp 2$ are now interchangeable in terms of probability distribution. Definition (\ref{FF}) of $F_1$ or $F_2$ then entails that $F_1(s_1,s_2) = F_2(s_2,s_1)$ for $(s_1,s_2) \in \mathbf{\Omega}$ and we denote by $F_0(s_1,s_2)$ the latter quantity; using similar arguments, we have $G_1 = G_2 = G$. By Proposition \ref{resol}.c, we further have $\psi_1(0) = \psi_2(0) = \lambda(1-\varrho)/2$ and function $J_1 = J_2 = J$ introduced in (\ref{J1J2}) is simply given by
\begin{equation}
J(s) = \frac{\lambda}{2}(1-\varrho)\frac{s}{s + \mu}.
\label{JJ}
\end{equation} 
Relations~\eqref{Fonctcomp} then specialize to the unique equation
\begin{equation}
\label{kernelsym}
K_1(s_1,s_2)F_0(s_1,s_2) + K_2(s_1,s_2)G(s_2) = J(s_2) + H(s_1,s_2),
\end{equation}
where general expression \eqref{Hbis} for $H$ now simply reduces to
\begin{equation}
H(s_1,s_2) = \frac{\lambda \mu(s_2 - s_1)}{2(\mu+s_1)(\mu+s_2)}M\left(\frac{s_1+s_2}{2} \right),
\label{Hsym}
\end{equation}
where
\begin{equation}
\label{Msym}
M(z) = G(2z+\mu)+F_0(2z+\mu,-\mu)
\end{equation}
(note the symmetry between transforms $F_1$ and $F_2$ mentioned above implies that $F_0(2z+\mu,-\mu) = F_0(-\mu,2z+\mu)$). 

Once function $H$ is expressed by (\ref{Hsym}) in terms of auxiliary function $M$, functional equation (\ref{kernelsym}) gives $F_1 = F_2 = F_0$ in terms of both $G$ and $M$. As univariate transform $G$ will be later shown to depend on function $M$ only, our remaining task is therefore to derive the latter function.

\subsection{Preliminary results}

Let us first assert some extension properties for analytic functions of interest. Recall from \cite[�3.3]{ravi} that the Laplace transform of the workload $\overline{U}_1$ in queue $\sharp 1$ when queue $\sharp 2$ has HoL priority is given by
\begin{equation}
\label{UHoL}
\E \big [e^{-s\overline{U}_1}  \big ] = \frac{2(1-\rho)s\xi^+(s)}{\lambda(1-b(s))(s-\xi^+(s))}
\end{equation}
for $\Re(s)\geq 0$, where $\xi^+(s)$ is the unique root of equation $\xi=K(\xi,s)$ which is positive for $s>0$. Specializing definition (\ref{Ker}) for $K(\xi,s)$ to the present symmetric case, equation $\xi = K(\xi,s)$ readily reduces to
\begin{equation}
\label{eqxi}
(s+\mu)\xi^2+\left(\mu^2-\frac{\lambda\mu}{2}+(\mu-\lambda)s\right)\xi-\frac{\lambda\mu}{2}s=0;
\end{equation}
its roots $\xi^+(s)$ and $\xi^-(s)$ are therefore given by
\begin{equation}
\label{xi+-}
\xi^{\pm}(s)  = \frac{-\left(\mu^2-\lambda\mu/2+(\mu-\lambda)s\right)\pm\sqrt{D(s)}}{2(s+\mu)}
\end{equation}
where discriminant $D(s) = (\mu^2-\lambda\mu/2+(\mu-\lambda)s)^2+2\lambda\mu s(\mu+s)$
is positive for $s\in \R\setminus]\zeta^-,\zeta^+[$ and non positive for $s\in [\zeta^-,\zeta^+]$, with
\begin{equation}
\label{zeta1+-}
\zeta^- = -\mu\frac{(\sqrt{\mu}+\sqrt{\lambda/2})^2}{\lambda/2+(\sqrt{\mu}+\sqrt{\lambda/2})^2}, \quad \quad \zeta^+ = -\mu \frac{(\sqrt{\mu}-\sqrt{\lambda/2})^2}{\lambda/2+(\sqrt{\mu}-\sqrt{\lambda/2})^2}.
\end{equation}
Functions $s \mapsto \xi^\pm(s)$ are defined for real $s \notin [\zeta^-,\zeta^+]$. With the convention $\sqrt{-1}=i$, we can define analytic or meromorphic extensions of these functions in the complex plane as follows. 

\begin{lemma}
Function $\xi^-$ (resp. $\xi^+$) can be analytically (resp. meromorphically) extended to the cut plane $\mathbb{C} \setminus [\zeta^-,\zeta^+]$.
\label{extendXI}
\end{lemma}
\begin{proof}
Function $s \mapsto \xi^+(s)$ is well-defined for $s \in \R \setminus ]\zeta^-,\zeta^+[$, whereas function $s \mapsto \xi^-(s)$ is well-defined for  $s\in \R \setminus ]\zeta^-,\zeta^+[$ and $s \neq -\mu$, with $\xi^-(-\mu) = \infty$. It is easily checked that for $s$ belonging to the vertical line $\Re(s) = (\zeta^-+\zeta^+)/2 <0$, we have $\Im(D(s)) = 0$ and 
$$
\Re(D(s)) =  -\frac{\lambda^3\mu^3}{8(\lambda^2+\mu^2) } - \Im(s)^2 < 0
$$
(note this vertical line and the real line are the only subsets of the complex plane on which $\Im(D(s)) = 0$). The Schwarz's reflection principle applied to function $\sqrt{D}$ with respect to the vertical line $\Re(s) = (\zeta^-+\zeta^+)/2$ then ensures that the function $E$ defined by $E(s) = - \sqrt{D(s)}$ for $\Re(s) \leq (\zeta^-+\zeta^+)/2$ and $E(s) = + \sqrt{D(s)}$ for $\Re(s) \geq (\zeta^-+\zeta^+)/2$ is globally analytic on the cut plane $\C\setminus [\zeta^-,\zeta^+]$. Let us then define functions $\xi$ and $\widetilde{\xi}$ by
\begin{equation}
\label{defxi}
\xi(s) = \left\{ \begin{array}{ll} 
\xi^-(s) \; \mbox{if } \Re(s)\leq \displaystyle \frac{\zeta^-+\zeta^+}{2}, \\ \\ 
\xi^+(s) \; \mbox{if } \Re(s)\geq \displaystyle \frac{\zeta^-+\zeta^-}{2},
\end{array}
\right.
\widetilde{\xi}(s) = \left\{ \begin{array}{ll}  
\xi^+(s) \; \mbox{if } \Re(s)\leq \displaystyle \frac{\zeta^-+\zeta^+}{2}, \\ \\ 
\xi^-(s) \; \mbox{if } \Re(s)\geq \displaystyle \frac{\zeta^++\zeta^-}{2},
\end{array} 
\right. 
\end{equation}
respectively. By construction, function $\xi$ is  a meromorphic extension of $\xi^+$ in $\C\setminus[\zeta^-,\zeta^+]$ with a pole at point $-\mu$, while function $\widetilde{\xi}$ is an analytic extension of $\xi^-$ in $\C \setminus [\zeta^-,\zeta^+]$.
\\
\indent
For notation simplicity, we will still denote by $\xi^+$ and $\xi^-$ their respective analytic continuation $\xi$ and $\widetilde{\xi}$ defined above.
\end{proof}

Consider now equation $s=K(s,s)$, whose unique non-zero solution is $-\mu(1-\varrho)$. As $s = K(s,s) \Leftrightarrow s = \xi^+(s) \; \mathrm{or} \; s = \xi^-(s)$, it is easily verified that solution $-\mu(1-\varrho)$ is associated with branch $\xi^+$ if $\varrho \geq 1/2$ and with branch $\xi^-$ if $\varrho \leq 1/2$. Define then
\begin{equation}
\label{s*sym}
\widetilde{s} = \left\{\begin{array}{ll} -\mu(1-\varrho) &\mbox{if } \; \varrho\geq 1/2,  
\\ \\ 
\zeta^+ & \mbox{if } \; \varrho \leq 1/2
\end{array} \right.
\end{equation}
(note that $\zeta^+ \leq -\mu(1-\varrho)$ for all $\varrho \in [0,1]$, as easily verified from the defining expression of polynomial $D(s)$ in (\ref{xi+-})).
\begin{lemma}
With the above notation, Laplace transform $G$ can be analytically extended to the half-plane $\widetilde{\omega} = \{s \in \mathbb{C} \; \mid \; \Re(s)>\widetilde{s}\}$; function $F_0$ can be analytically extended to  $\widetilde{\mathbf{\Omega}} = \{(s_1,s_2) \in \mathbb{C}^2 \; \mid \; \Re(s_1)>\max(\widetilde{s},\widetilde{s}-\Re(s_2))\}$.
\label{extendGF0}
\end{lemma}
\begin{proof}
By (\ref{UHoL}) and Lemma \ref{extendXI}, transform $s \mapsto \E(e^{-s \overline{U}})$ is analytic for $\Re(s) >\widetilde{s}$. This transform may have a pole only at any point $s$ such that $\xi^+(s) = s$. By the above discussion, we actually have a pole at $\widetilde{s} = -\mu(1-\varrho)$ when $\varrho>1/2$; it is not a pole when $\varrho\leq 1/2$ but the algebraic singularity at point $\widetilde{s} = \zeta^+$ instead occurs. Applying then Corollary~\ref{extensions} with $\widetilde{s}_1 = \widetilde{s}_2 = \widetilde{s}$, the extended analyticity domains for $G$ and $F_0$ follow. 
\end{proof}
Following definition (\ref{Msym}) and Lemma \ref{extendGF0}, function $M$ is consequently analytic on the half-plane $v_M = \{z \in \mathbb{C} \; \mid \; \Re(z)>\tilde{s}/2\}$.

\subsection{The cubic equation}

As detailed in Section \ref{SQFqueue}, the final determination of function $M$ relies on the algebraic and analytic properties for the branches of a cubic polynomial equation.
\begin{prop}
\textbf{a)} For given $z > 0$ and $z^* > 0$, relations
\begin{equation}
z = \frac{s+\xi^-(s)}{2}, \; \; z^* = \frac{s+\xi^+(s)}{2}
\label{z-s}
\end{equation}
can be inverted in variable $s$ as
\begin{equation}
s = z - \alpha(z), \; \; s = z^* - \beta(z^*)
\label{z-sinvert}
\end{equation} 
respectively, where $\alpha(z)$ and $\beta(z)$ are the two non positive  roots of cubic equation $R(w,z) = 0$ in variable $w$, with
\begin{equation}
\label{cubiceq}
R(w,z) = w^3-(\lambda-z)w^2-(z+\mu)^2w-z(z+\mu)(z+\mu-\lambda).
\end{equation}

For $z > 0$, $\xi^-(s)$ and $\xi^+(s)$ are given by $\xi^-(s) = z+\alpha(z)$ and $\xi^+(s) = z^*+\beta(z^*)$.
\indent
\textbf{b)} For $z \geq 0$, cubic polynomial $R(w,z)$ has three distinct real roots $\alpha(z)$, $\beta(z)$ and $\gamma(z)$ such that $\alpha(z)<-z\leq \beta(z)\leq 0<\gamma(z)$ and $\beta(0)=0$.
\label{Roots-R}
\end{prop}
\begin{proof}
\textbf{a)} Eliminating $\xi^-(s)$ between first relation (\ref{z-s}) and polynomial equation (\ref{eqxi}) satisfied by $\xi^-(s)$, we can write $s = z - \alpha(z)$ where $R(\alpha(z),z) = 0$, cubic polynomial $R(w,z)$ being defined as in (\ref{cubiceq}). 
Similarly, eliminating $\xi^+(s)$ between second relation (\ref{z-s}) and equation (\ref{eqxi}) enables us to write $s = z^* - \beta(z^*)$ where $R(\beta(z^*),z^*) = 0$ with identical polynomial $R(w,z)$. 

We readily deduce, in particular, that $\xi^-(s) = 2z-s=z+\alpha(z)$, and similarly  $\xi^+(s) = 2z^* - s = z^* + \beta(z^*)$.

\textbf{b)}  For $z > 0$, we have $R(-z,z) = \lambda\mu z > 0$ and $R(0,z) = -z(z+\mu)(z+\mu-\lambda) < 0$ since $\lambda < \mu$ by the stability condition. Further accounting for its values at infinity, we deduce that cubic polynomial $R(w,z)$ has three real roots for $z \geq 0$; denoting them by $\alpha(z)$, $\beta(z)$ and $\gamma(z)$, the latter discussion implies the claimed inequalities.

We finally verify that roots $\alpha(z)$ and $\beta(z)$ previously characterised either in \textbf{a)} or \textbf{b)} actually coincide. In fact, let $z>0$ so that $z = (s+\xi^-(s))/2$; given the variations of the function $s \mapsto \xi^-(s)$ for $s > \zeta^+$, $s$ has to be sufficiently large for $z = (s+\xi^-(s))/2$ to be positive; this implies that we necessarily have $s=z-\alpha(z)$ where $\alpha(z)$ is the smallest root of polynomial $R(w,z)$. We can similarly prove that if $z^* = (s+\xi^+(s))/2>0$, then $s=z^*-\beta(z^*)$ where $\beta(z^*)$ is the second smallest root of $R(w,z)$.
\end{proof}

As solutions to a polynomial equation, algebraic functions $z \mapsto \alpha(z)$, $z \mapsto \beta(z)$ and $z \mapsto \gamma(z)$ can be analytically defined in $\C$ cut along some slits. Specifically, writing $R(z,w)$ as $R(z,w) = w^3+R_1(z)w^2+R_2(z)w+R_3(z)$ with coefficients $R_1(z)$, $R_2(z)$ and $R_3(z)$ defined by (\ref{cubiceq}) and introducing 
$$
\widetilde{P}(z) = R_2(z)-\frac{R^2_1(z)}{3}, \quad \widetilde{Q}(z)=R_3(z)-\frac{R_1(z)R_2(z)}{3}+\frac{2R_1^3(z)}{27},
$$
any solution $\epsilon(z) \in \{\alpha(z), \beta(z), \gamma(z)\}$ to $R(w,z) = 0$ can be expressed by Cardano's formula \cite[p.~16]{Cox} as
\begin{equation}
\label{Cardan}
\epsilon = -\frac{R_1}{3} + j^m \sqrt[3]{\frac{1}{2} \left ( - \widetilde{Q}+ \sqrt{- \frac{\Delta}{27} } \right )} + j^n \sqrt[3]{\frac{1}{2} \left ( - \widetilde{Q} - \sqrt{- \frac{\Delta}{27} } \right )},
\end{equation}
where $j = e^{2i\pi/3}$, the pair $(m,n)$ can take either value $(0,0)$, $(1,2)$ or $(2,1)$, and with discriminant $\Delta$ defined by $\Delta(z) = -4\widetilde{P}(z)^3-27\widetilde{Q}(z)^3$. Some algebra shows that discriminant $\Delta(z)$ factorizes as $\Delta(z) = (z+\mu)\delta(z)$ with 
\begin{align}
\delta(z) = 16(\lambda^2+\mu^2)z^3 & - (16\lambda^3-24\lambda^2\mu+24\lambda\mu^2-32\mu^3)z^2 
\nonumber \\ & +(4\lambda^4-4\lambda^3\mu+21\lambda^2\mu^2-20\lambda\mu^3+20\mu^4)z+\lambda^2\mu^3+4\mu^5.
\label{smalldelta}
\end{align}
The respective analyticity domains of functions $\alpha$, $\beta$ and $\gamma$ are  related to the roots of discriminant $\Delta(z)$, these roots defining the so-called ramification points for such algebraic functions.

\begin{lemma} 
\textbf{a)} Discriminant $\Delta(z)$ has four distinct roots, namely two real roots $\eta_1 \in \; ]-\mu,0[$ and $\eta_2 = -\mu$ and two complex conjugate roots $\eta_3$ and $\eta_4$. 
\\
\indent
\textbf{b)} Algebraic functions $\alpha$, $\beta$ and $\gamma$ are analytic on the cut plane $\C \setminus [\eta_2,\eta_1]$, $\C \setminus ([\eta_2,\eta_1]\cup[\eta_3,\eta_4])$ and $\C \setminus [\eta_3,\eta_4]$, respectively.
\label{discriminant}
\end{lemma}

\begin{proof}
\textbf{a)} The point $\eta_2=-\mu$ is clearly a root of $\Delta(z) = (z + \mu)\delta(z)$ and it is simple since $\delta(-\mu)=-4\lambda\mu(\lambda+\mu)^3 \neq 0$ in view of expression (\ref{smalldelta}). Moreover, as the coefficient of the leading term of the cubic polynomial $\delta(z)$ is positive, as $\delta(0) > 0$  and $\delta(-\mu)<0$, discriminant $\Delta(z)$ has at least another negative real root $\eta_1$ between $-\mu$ and 0.

Besides, the discriminant of $\delta(z)$ is easily calculated as $E = -2^{7}\mu^{14}\varrho E_0(\varrho)^3$ with $E_0(\varrho) = 4\varrho^4-2\varrho^3+15\varrho^2-2\varrho+4$; as $E_0(\varrho) > -2 - 2 + 4 = 0$ for $0 < \varrho < 1$, we have $E < 0$. It then follows from \cite[Theorem 1.3.1]{Cox} that cubic polynomial $\delta(z)$ with real coefficients has only one real root, namely $\eta_1$, the two others $\eta_3$ and $\eta_4$ being complex conjugates.

\textbf{b)} By considering the analytic continuation of function $\sqrt{\Delta}$ such that $\sqrt{\Delta(0)}>0$ in $\C \setminus ([\eta_2,\eta_1]\cup[\eta_3,\eta_4])$, formulas (\ref{Cardan}) enable us to analytically continue function $\alpha$ to the cut plane $\C \setminus [\eta_2,\eta_1]$, function $\beta$ to the cut plane $\C \setminus ([\eta_2,\eta_1]\cup[\eta_3,\eta_4])$ and function $\gamma$ to the cut plane $\C \setminus [\eta_3,\eta_4]$, respectively.
\end{proof}

The graphs of functions $X^-:s \to (s+\xi^-(s))/2$ and $X^+:s \to (s+\xi^+(s))/2$ are illustrated in Fig.~\ref{courbexi} on interval $[\zeta^+,0]$. Function $X^+$ is increasing while function $X^-$ reaches its minimum at some point $s^*$; 
$X^-$ is decreasing on interval $]\zeta^+,s^*[$ and increasing on interval $]s^*,0[$. Recall from Proposition \ref{Roots-R} that $s=z-\alpha(z)$ entails $z= X^-(s)$;  conversely, we have $z= X^-(s)$ for $s \in [s^*,+\infty[$. Function $z \mapsto \alpha(z)$ is thus defined and regular for $z \in \; ]\eta_1,+\infty[$ where 
$\eta_1 = X^-(s^*) = (s^*+\xi^-(s^*))/2$. Using similar arguments, $z \mapsto \beta(z)$ is shown to be regular for $z > X^\pm(\zeta^+)$.

\begin{figure}[t]
\centering
\scalebox{.40}{\includegraphics{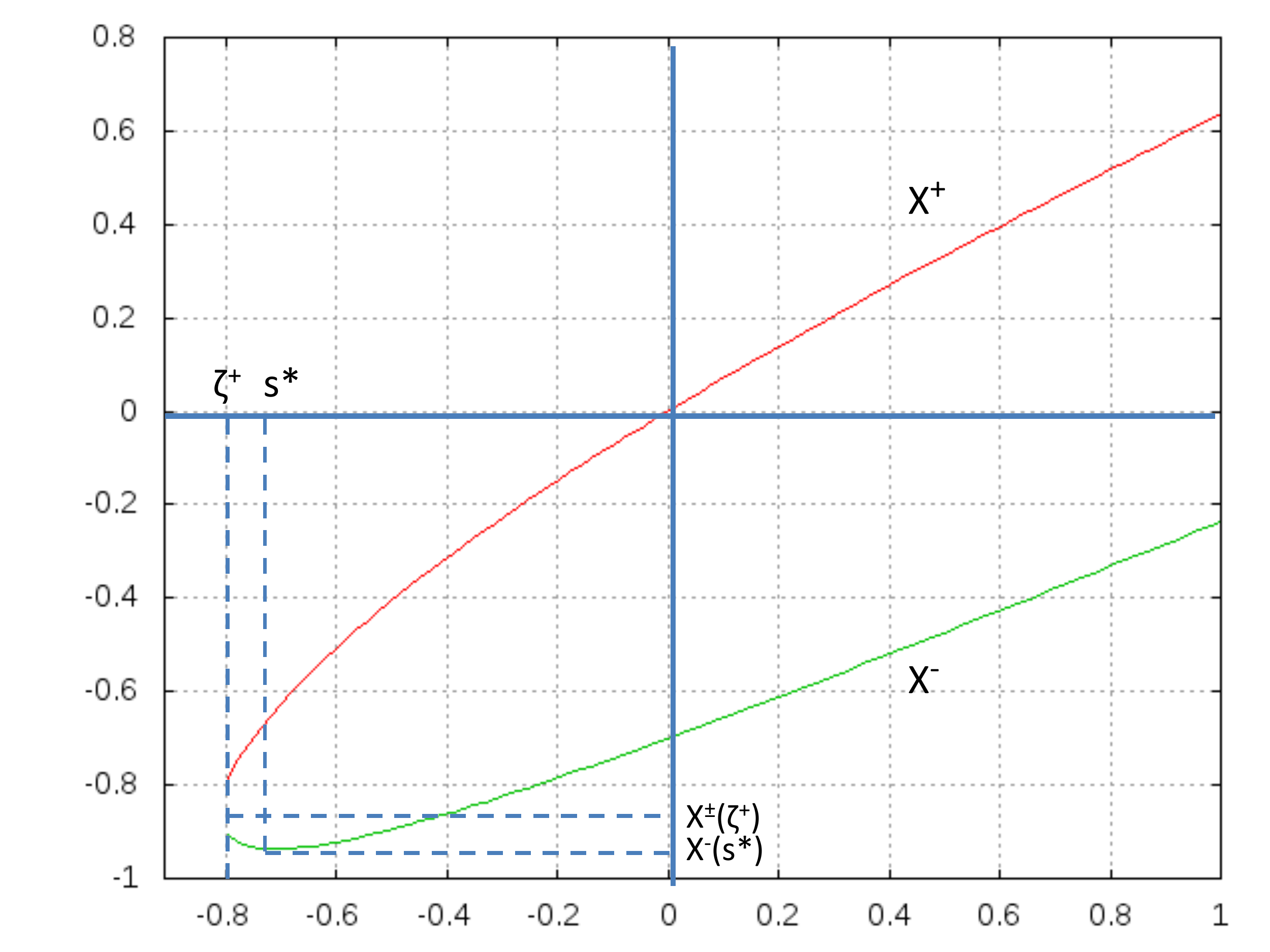}}
\caption{Graphs of functions $X^\pm$ (for $\lambda=1.2$, $\mu=2$).\label{courbexi}}
\end{figure}


\section{The SQF queue in the symmetric case}
\label{SQFqueue}


On the basis of the preliminary results obtained in Section \ref{EXPI}, we are now ready to provide a final solution for auxiliary function $M$ (Section \ref{RSE}) and determine an extended analyticity domain (Section \ref{SS}), from which all relevant probabilistic properties for the symmetric queue can be derived (Sections \ref{EQP} and \ref{LQA}). 

\subsection{Real series expansion}
\label{RSE}


We first provide a series expansion for Laplace transform $G$ on some real interval of its definition domain. The proposition below states the core functional equation verified by function $M$.
\begin{prop}
Function $M$ defined by (\ref{Msym}) verifies the functional equation
\begin{equation}
M(z) = q(z) \cdot M \circ h(z) + L(z)
\label{FonctMsym}
\end{equation}
for $z > 0$, with
\begin{equation}
\left\{
\begin{array}{ll}
q(z) = \displaystyle \frac{\mu+\xi^-(s)}{\mu+\xi^+(s)},
\\ \\
L(z) = \displaystyle  (1-\varrho)\frac{s(\xi^+(s)-\xi^-(s))}{(s-\xi^+(s))(s-\xi^-(s))}\frac{\mu + \xi^-(s)}{\mu},
\\ \\
h(z) = \displaystyle \frac{\xi^+(s)+s}{2},
\label{defqLh}
\end{array} \right.
\end{equation}
where $s = z - \alpha(z)$ is the unique solution to equation $s + \xi^-(s) = 2z$.
\label{C1}
\end{prop}

\begin{proof}
Applying equation \eqref{FonctG2} successively to points $(s_1,s_2) = (\xi^+(s),s)$ and $(s_1,s_2) = (\xi^-(s),s)$ with identical ordinate $s$, we obtain
\begin{equation}
\left\{
\begin{array}{ll}
(s-\xi^+(s))G(s) = \displaystyle \frac{\lambda(1-\varrho)s}{2(s+\mu)} + \displaystyle  \frac{\lambda \mu(s - \xi^+(s))}{2(\mu+\xi^+(s))(\mu+s)}M\left( \frac{\xi^+(s)+s}{2} \right),
\\
\\
(s-\xi^-(s))G(s) = \displaystyle \frac{\lambda(1-\varrho)s}{2(s+\mu)} + \displaystyle  \frac{\lambda \mu(s - \xi^-(s))}{2(\mu+\xi^-(s))(\mu+s)}M\left( \frac{\xi^-(s)+s}{2} \right)
\end{array} \right.
\label{MSymequ}
\end{equation}
after using expression (\ref{JJ}) for $J(s)$ and formula (\ref{Hsym}) for $H(s_1,s_2)$; equations (\ref{MSymequ}) hold for sufficiently large $s$ so that $2z = \xi^-(s) + s$ is positive. Using the fact that  $s = z - \alpha(z)$. Equating the common value of $G(s)$ from (\ref{MSymequ}) and using the fact that $h(z) = (\xi^+(s)+s)/2$ gives functional equation (\ref{FonctMsym}).
\end{proof}
By Proposition \ref{Roots-R}.a, $\xi^-(s) = z + \alpha(z)$ depends on the branch $\alpha(z)$ only. As $\xi^+(s)\xi^-(s) = -\lambda\mu/2(s+\mu)$ in view of defining equation (\ref{eqxi}), definition (\ref{defqLh}) for $h(z)$ further gives
\begin{equation}
h(z) = \frac{\xi^+(s)+s}{2} = \frac{z-\alpha(z)}{2}\left[1-\frac{\lambda\mu}{2(z+\alpha(z))(z-\alpha(z)+\mu)}  \right],
\label{hz}
\end{equation} 
and a similar rational expression is derived from (\ref{defqLh}) for $L(z)$ in terms of $\alpha(z)$. As a consequence, given functions $q$, $L$, and $h$ depend only on the branch $\alpha$ of cubic equation $R(w,z) = 0$. Note also that by the notation introduced in inversion relations (\ref{z-s})-(\ref{z-sinvert}), $h(z)$ just coincides with $z^*$; the mapping $z \mapsto h(z) = z^*$ is now introduced in view of its iterated composition, as will be shown in the central result below.

\begin{theorem}
The Laplace transform $G$ can be expressed as
\begin{equation}
G(s) = \frac{\lambda}{2(s+\mu)} \left[ \frac{s(1-\varrho)}{s - \xi^-(s)} + \frac{\mu}{\mu + \xi^-(s)} M(z) \right]
\label{Gsym}
\end{equation}
for sufficiently large real $s$ so that $z = (\xi^-(s)+s)/2 > 0$ and where $M(z)$ is given by the series expansion
\begin{equation}
M(z) = \sum_{k=0}^{+\infty} \; \prod_{\ell = 0}^{k-1} q(h^{(\ell)}(z)) \cdot L(h^{(k)}(z))
\label{Mseries}
\end{equation}
with functions $q$, $L$ and $h$ defined by \eqref{defqLh}, and $h^{(k)} = h \circ .... \circ h$ denoting the $k$-th iterate of function $h$.
\label{resolsym}
\end{theorem}

\begin{proof}
Iterating functional equation (\ref{FonctMsym}) for $z>0$ yields
\begin{equation}
M(z) = \sum_{k=0}^{K} \; \prod_{\ell = 0}^{k-1} q(h^{(\ell)}(z)) \cdot L(h^{(k)}(z)) + E^{(K)}(z)
\label{quotientbis}
\end{equation}
(the product being equal to 1 for $k = 0$), with remainder
$$
E^{(K)}(z) = \prod_{k=0}^{K} q(h^{(k)}(z)) \cdot M(h^{(K+1)}(z)).
$$
To show that $E^{(K)}(z) \rightarrow 0$ as $K \uparrow +\infty$, let us fix some $z>0$. We first prove that the sequence $z^{(k)} =  h^{(k)}(z)$, $k \geq 0$, is strictly increasing and tends to $+\infty$ when $k \uparrow +\infty$. In fact, 
as $\alpha(z)<-z$ for $z>0$ by Proposition \ref{Roots-R}.b, we deduce from expression (\ref{hz}) that $h(z)>z$ for $z > 0$ and the sequence $z^{(k)} = h^{(k)}(z)$, $k \geq 0$, is thus strictly increasing. Moreover, if that sequence were upper bounded, it would tend to a finite limit $z_\infty$ such that $h(z_\infty) = z_\infty$ and the number  $s_\infty=z_\infty-\alpha(z_\infty)$ is positive; but using expression (\ref{hz}) for $h(z)$, equality $h(z_\infty) = z_\infty$ reduces to
$$
z_\infty = \frac{s_\infty}{2} \left [ 1 - \frac{\lambda\mu}{2(2z_\infty-s_\infty)(s_\infty+\mu)} \right ],
$$
or equivalently
$$
(2z_\infty - s_\infty)^2 = - \frac{\lambda\mu s_\infty}{2(s_\infty + \mu)},
$$
and the latter would define a simultaneously positive and negative quantity, a contradiction. We thus conclude that $z^{(k)} \to +\infty$ when $k \uparrow +\infty$.

Besides, we derive from definition (\ref{defqLh}) for $q$ that $\lim_{z \uparrow +\infty} q(z) = r$, where
$$
r = \frac{\lambda+\mu-\sqrt{\lambda^2+\mu^2}}{\lambda+\mu+\sqrt{\lambda^2+\mu^2}} < 1.
$$
By definition (\ref{Msym}) of function $M$, the sequence $M(h^{(k)}(z)) = M(z^{(k)})$, $k \geq 0$, is bounded since both $G$ and $F_0$ vanish at infinity as Laplace transforms of regular densities. It follows that remainder $E^{(K)}(z)$ is  $O \big [ M(h^{(K+1)}(z))r^K \big ] = O(r^K)$ and therefore tends to 0 as $K \uparrow +\infty$. The finite sum in (\ref{quotientbis}) thus converges as $K \uparrow +\infty$.

Formula (\ref{Gsym}) for $G(s)$ eventually follows from the latter expansion inserted into second  equation (\ref{MSymequ}). 
\end{proof}


\subsection{Analytic extension}
\label{SS}


We now specify the smallest singularity of Laplace transform $G$; to this end, we first deal with the analyticity domain of auxiliary function $M$. Recall by definition (\ref{Msym}) that $M$ is known to be analytic at least in the half-plane $v_M = \{z \in \mathbb{C} \; \mid \; \Re(z)>\tilde{s}/2\}$, where $\tilde{s}$ is defined by \eqref{s*sym}.

\begin{prop} 
Function $M$ can be analytically continued to the half-plane $V_M $ (with $v_M \varsubsetneq V_M$) defined by
\begin{itemize}
\item[\textit{\textbf{a)}}] $V_M = \Big\{ z \in \mathbb{C} \mid \; \Re(z) > \frac{1}{2}\left( \sigma_0 - \frac{\mu}{2} \right) \Big\}$ in case $\varrho > 1/2$, where we set $\sigma_0 = -\mu(1-\varrho)$;
\item[\textit{\textbf{b)}}] $V_M = \Big\{z \in \mathbb{C} \mid \; \Re(z) > \eta_1 \Big\}$ 
in case $\varrho \leq 1/2$, where $\eta_1 < 0$ is the largest real root of discriminant $\Delta(z)$.
\end{itemize}
\label{extendM}
\end{prop}
The proof of Proposition \ref{extendM} is detailed in Appendix \ref{A4}. We now turn to transform $G$ and determine its singularities with smallest module. Recall by Corollary \ref{extensions} that $G$ has no singularity in $\{s \in \mathbb{C} \; \mid \; \Re(s) > \widetilde{s}\}$.

\begin{theorem}
\label{domainG}
The singularity with smallest module of transform $G$ is
\begin{itemize}
\item[\textbf{a)}] For $\varrho > 1/2$, a simple pole at $s =\sigma_0 = -\mu(1-\varrho)$ with leading term
\begin{equation}
G(s) \sim \frac{r_0}{s-\sigma_0}
\label{resGsym1}
\end{equation}
with $r_0 = \mu(1-\varrho)(2\varrho - 1)/4$;	
\item[\textbf{b)}] For $\varrho < 1/2$, an algebraic singularity at $s = \zeta^+$ with leading term
\begin{equation}
G(s) - G(\zeta^+) \sim r^+(s-\zeta^+)^{1/2}
\label{resGsym2}
\end{equation}
at first order in $\sqrt{s-\zeta^+}$, where factor $r^+$ is given by
$$
r^+ = \frac{\lambda \sqrt{ (\lambda^2 + \mu^2)(\zeta^+ - \zeta^-)}}{4(\zeta^+ +\mu)^2} \left [ \frac{\zeta^+(1-\varrho)}{(\zeta^+ - a^+)^2} - \frac{\mu M(z^+)}{(\mu + a^+)^2} \right ]
$$
with constants $a^+ = -\mu + \sqrt{\lambda\mu/2}$, $2z^+ = \zeta^+ + a^+$ and where $\zeta^+$, $\zeta^-$ are given in (\ref{zeta1+-}).
\end{itemize}
\label{Gsing}
\end{theorem}

\begin{proof}
Consider again the two following cases:
\\
\indent
\textbf{a)} if $\varrho > 1/2$, write the 1st equation (\ref{MSymequ}) as
\begin{equation}
G(s) = \frac{\lambda}{2(s+\mu)} \left[ \frac{s(1-\varrho)}{s - \xi^+(s)} + \frac{\mu}{\mu + \xi^+(s)} M \circ h(z) \right];
\label{Gsymbis}
\end{equation}
as $s \rightarrow \sigma_0$, we have $\xi^+(s) \rightarrow \sigma_0$ while $h(z) = (s + \xi^+(s))/2 \rightarrow \sigma_0$. Proposition \ref{extendM} then ensures that $M \circ h$ is analytic at $z = \sigma_0$ since $(\sigma_0 - \mu/2)/2 < \sigma_0$ for $\varrho > 1/2$. As $G(s)$ has no singularity for $\Re(s) > \sigma_0$, we conclude from expression (\ref{Gsymbis}) that $G$ has a simple pole at $s = \sigma_0$ with residue
$$
r_0 = \frac{\lambda}{2(\sigma_0 + \mu)} \left [ \frac{\sigma_0(1-\varrho)}{1-\xi^+{'} (\sigma_0)} \right ]
$$
where $\sigma_0 = -\mu(1-\varrho)$. Differentiating formula ~\eqref{xi+-} for $\xi^+(s)$ at $s = \sigma_0$, we further calculate $\xi^{+}{'}(\sigma_0) = 1/(2\varrho - 1)$; residue $r_0$ in leading term (\ref{resGsym1}) then follows;
\\
\indent 
\textbf{b)} if $\varrho < 1/2$, let $s \rightarrow \sigma_0$ so that $\xi^+(s) \rightarrow -\mu/2$ and $h(z) = (s + \xi^+(s))/2 \rightarrow z_0$ where $2z_0 = \sigma_0 - \mu/2$. Proposition \ref{extendM} then ensures that $M$ is analytic at $z = z_0$ since $z_0 > \eta_1$. We conclude from expression (\ref{Gsymbis}) and the latter discussion that $\sigma_0$ is not a singularity of $G$. 
\\
\indent
By definition~\eqref{xi+-} of $\xi^+(s)$, where $D(s)$ is factorized as $D(s) = D_0(s-\zeta^-)(s-\zeta^+)$ with $D_0 = (\mu-\lambda)^2+2\lambda\mu = \lambda^2 + \mu^2$, we obtain
$$
\xi^+(s) = a^+ + E_0\sqrt{s-\zeta^+} + o(s-\zeta^+)^{1/2},
$$
where $a^+ = -\mu + \sqrt{\lambda\mu/2}$ and with constant $E_0=\sqrt{D_0(\zeta^+-\zeta^-)}/2(\mu+\zeta^+)$. By expression (\ref{Gsymbis}) for $G(s)$, we then obtain
\begin{align}
G(s) & = \frac{\lambda}{2(\zeta^+ + \mu)} \left [ \frac{\zeta^+(1-\varrho)}{\zeta^+-a^+-E_0 \sqrt{s-\zeta^+}+ ...} + \frac{\mu M(z^+)}{\mu+a^+ + E_0 \sqrt{s-\zeta^+}+ ...}\right ]
\nonumber \\
& = G(\zeta^+) + r^+(s-\zeta^+)^{1/2} + ...
\nonumber
\end{align} 
since $z = (s + \xi^-(s))/2 \rightarrow z^+$ with $z^+$ defined in (\ref{resGsym2}). Expansion (\ref{resGsym2}) then follows with associated factor $r^+$; we conclude that the singularity with smallest module of $G$ is $\zeta^+$, an algebraic singularity with order 1.
\end{proof}


\subsection{Empty queue probability}
\label{EQP}


The  results obtained in the previous section enable us to give a closed-form expression for the empty queue probability in terms of auxiliary function $M$ only.
\begin{prop}
Probability $G(0) = \mathbb{P}(U_1 > 0, U_2 = 0)$ is given by
\begin{equation}
G(0) =  M \left ( \frac{\lambda - 2\mu}{4}\right )
\label{G(0)}
\end{equation}
with $M$ given by series expansion (\ref{Mseries}).
\label{G0sym}
\end{prop}

\begin{proof}
Apply relation (\ref{Gsymbis}) for $G(s)$ with $s = 0$; as $\xi^+(0) = 0$, we then derive that $h(z) = (0 + \xi^+(0))/2 = 0$ hence
\begin{equation}
G(0) = \frac{\lambda}{2\mu} \left [ \frac{1-\varrho}{1-\xi^{+}{'}(0)} + M(0) \right ];
\label{G(0)bis}
\end{equation}
differentiating formula ~\eqref{xi+-} for $\xi^+(s)$ at $s = 0$ gives $\xi^+{'}(0) = \varrho/(\varrho - 2)$ so that the first term inside brackets in (\ref{G(0)bis}) reduces to $(1-\varrho)(1 - \varrho/2)$. Now, applying (\ref{FonctMsym}) to value $s = 0$ (with corresponding pair $z = (s + \xi^-(s))/2 = (\lambda-2\mu)/4$ and $h(z) = (s + \xi^+(s))/2 = 0$) shows that the right-hand side of (\ref{G(0)bis}) also equals $M(z) = M((\lambda-2\mu)/4)$, as claimed. 
\end{proof}
By (\ref{P0}) and (\ref{G(0)}), we derive the probability 
$$
\mathbb{P}(U_1 = 0) = 1 - \varrho + G(0)
$$
that either queue $\sharp 1$ or $\sharp 2$ is empty. 

We depict in Figure~\ref{FigG0} the variations of $\mathbb{P}(U_1 = 0)$ in terms of load $\varrho = \lambda/\mu$ when fixing $\mu = 1$ (for comparison, the black dashed line represents the empty queue probability $\mathbb{P}(U = 0) = 1 - \varrho$  for the unique queue aggregating all jobs from either class $\sharp 1$ or $\sharp 2$). The numerical results show that $\mathbb{P}(U_1 = 0)$ decreases to a positive limit, approximately $0.251...$, when $\varrho$ tends to 1; this can be interpreted by saying that, while the global system is unstable and sees excursions of either variable $U_1$ or $U_2$ to large values, one of the queues remains less than the other for a large  period of time and has therefore a positive probability to be emptied by the server.

\begin{figure}[hbtp]
\scalebox{1}{\includegraphics[width=15cm, trim = 50 500 0 50 cm,clip]{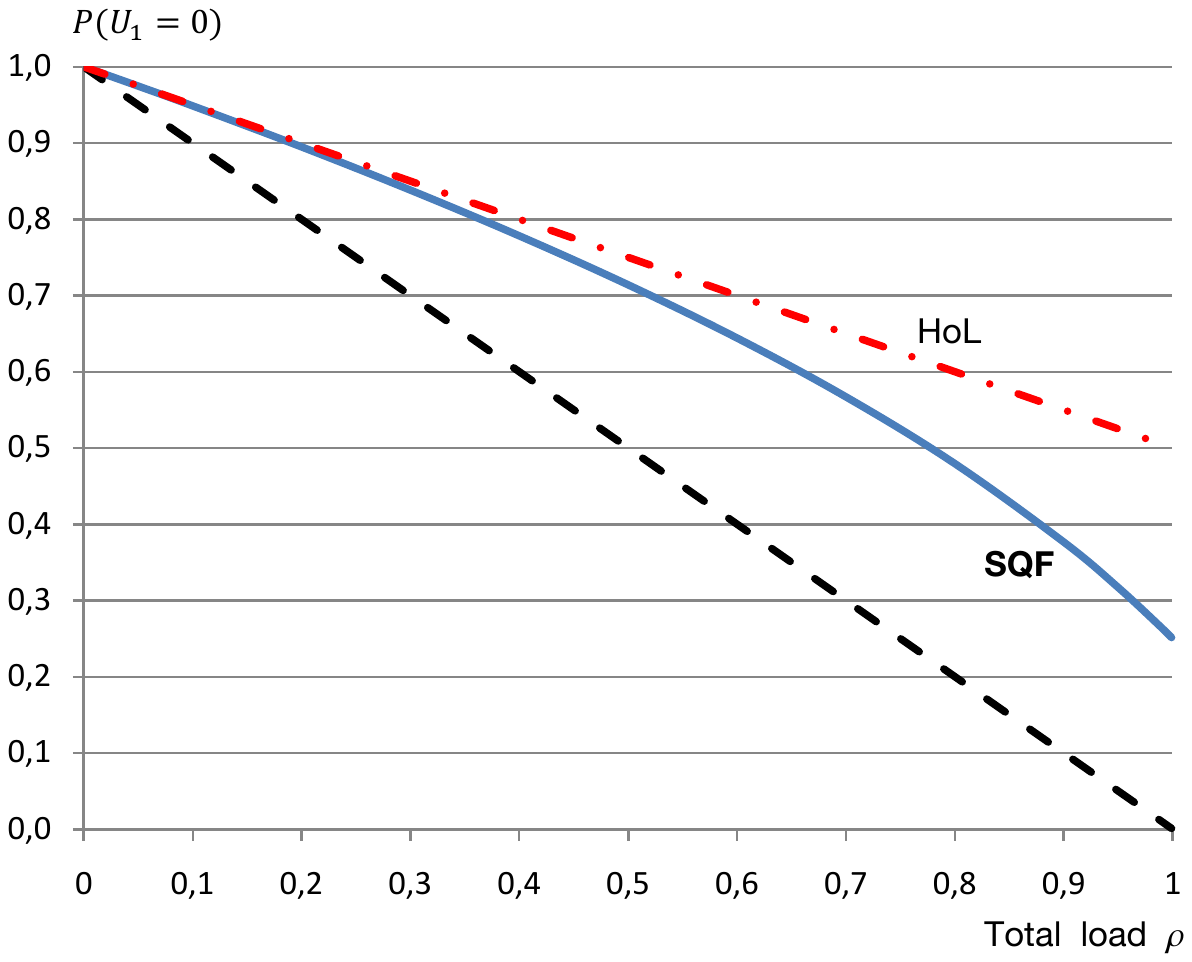}}
\caption{\textit{For the symmetric case, empty queue probability $\mathbb{P}(U_1 = 0) = \mathbb{P}(U_2 = 0)$ as a function of total load $\varrho$ (fixing $\mu = 1$).}}
\label{FigG0}
\end{figure}

Furthermore, the red dashed line depicts the empty queue probability 
\begin{equation}
\mathbb{P}(\underline{U}_1 = 0) = 1 - \varrho_1 = 1 - \frac{\varrho}{2}
\label{HoLP0}
\end{equation}
if the server were to apply a  preemptive HoL policy with highest priority given to queue $\sharp 1$; following lower bound (\ref{stochlowup}), we have $\mathbb{P}(U_1 = 0) \leq \mathbb{P}(\underline{U}_1 = 0)$. We further notice that for $\varrho \uparrow 1$, the positive limit of $\mathbb{P}(U_1 = 0)$ derived above for SQF is close enough to the maximal limit $0.5$ of $\mathbb{P}(\underline{U}_1 = 0)$. The above observations consequently show that the SQF policy compares favorably to the optimal HoL policy by guaranteeing a non vanishing empty queue probability for each traffic class at high load.


\subsection{Large queue asymptotics}
\label{LQA}


We finally derive asymptotics for the distribution of workload $U_1$ or $U_2$ in either queue, i.e., the estimates of tail probabilities $\mathbb{P}(U_1 > u)$ for large queue content $u$. We shall invoke the following Tauberian theorem relating the singularities of a Laplace transform to the asymptotic behavior of its inverse \cite[Theorem 25.2, p.237]{Doe}.

\begin{theorem} Let $F$ be a Laplace transform and $\omega$ be its singularity with smallest module, with $F(s) \sim \kappa_0 (s - \omega)^{\nu_0}$ as $s \rightarrow \omega$ for $\kappa_0 \neq 0$ and $\nu_0 \notin \mathbb{N}$ (replace $F$ by $F - F(\omega)$ if $F(\omega)$ is finite). The Laplace inverse $f$ of $F$ is then estimated by 
$$
f(u) \sim \frac{\kappa_0}{\Gamma(-\nu_0)} \frac{ e^{\omega u}}{u^{\nu_0 + 1}}
$$
for $u \uparrow +\infty$, where $\Gamma$ denotes Euler's $\Gamma$ function.
\label{Tauberian}
\end{theorem}
\noindent
Note that the fact that $F(\omega)$ is finite or not does not change the estimate of inverse $f$ at infinity. Before using that theorem for the tail behavior of either $U_1$ or $U_2$, we first state some simple bounds for their distribution tail. 


The global workload $U = U_1 + U_2$ is identical to that in an $M/M/1$ queue with arrival rate $\lambda$ and service rate $\mu$. The complementary distribution function of $U$ is therefore given by $\mathbb{P}(U > u) = \varrho e^{\sigma_0 u}$
for all $u \geq 0$, with $\sigma_0 = -\mu(1-\varrho)$; the distribution tail of workload $U_1$ or $U_2$ therefore decreases at least exponentially fast at infinity.

Following upper bound (\ref{stochlowup}) relating $U_1$ to variable $\overline{U}_1$ corresponding to a HoL service policy with highest priority given to queue $\sharp 2$, we further have 
\begin{equation}
\mathbb{P}(U_1 > u_1) \leq \mathbb{P}(\overline{U}_1 > u_1)
\label{Rough1}
\end{equation}
for all $u_1 \geq 0$. The Laplace transform of $\overline{U}_1$ is given by Equation~\eqref{UHoL} and is meromorphic in the cut plane $\mathbb{C} \setminus [\zeta^-,\zeta^+]$, with a possible pole at $\sigma_0 = -\mu(1-\rho)$.  Specifically, the application of Theorem \ref{Tauberian} shows that the tail behavior of $\overline{U}_1$ is given by
\begin{equation}
\mathbb{P}(\overline{U}_1 > u) = 
\left\{
\begin{array}{lll}
O(e^{\sigma_0 u}) & \; \mbox{ if } \; \varrho > \displaystyle \frac{1}{2},
\\
\\
\displaystyle O\left(\frac{e^{-\frac{\mu}{2} u}}{\sqrt{u}}\right) & \; \mbox{ if } \; \varrho = \displaystyle \frac{1}{2},
\\
\\
\displaystyle O\left(\frac{e^{\zeta^+ u}}{u^{3/2}}\right) & \; \mbox{ if } \; \varrho < \displaystyle \frac{1}{2},
\end{array} \right.
\label{asymptHoL}
\end{equation} 
for large $u$. The tail behavior of $\overline{U}_1$, and therefore $U_1$, may therefore be either exponential or subexponential according to system parameters. We precisely have the following result.

\begin{theorem} The workload in queue $\sharp 1$ is such that
\begin{equation}
\mathbb{P}(U_1 > u) \sim \left\{
\begin{array}{lll}
\displaystyle \left ( \varrho - \frac{1}{2} \right ) \cdot e^{\sigma_0u} & \mathrm{if} & \varrho > \displaystyle \frac{1}{2},
\\
\\
 \displaystyle \frac{1}{\sqrt{2\pi}} \cdot \frac{e^{-\frac{\mu}{2} u}}{\sqrt{\mu u}} & \mathrm{if} & \varrho = \displaystyle \frac{1}{2},
\\
\\
\displaystyle \displaystyle \kappa \cdot \frac{e^{\zeta^+ u}}{u^{3/2}} & \mathrm{if} & \varrho < \displaystyle \frac{1}{2},
\end{array} \right. 
\label{asymptsym}
\end{equation}
for large $u$, with constants $\sigma_0 = -\mu(1-\varrho)$ and 
$$
\kappa = \frac{(\zeta^+ + \mu)r^+}{\zeta^+\lambda \sqrt{\pi}},
$$
where $\zeta^+ \leq \sigma_0$ is given by (\ref{zeta1+-}) and $r^+$ by  (\ref{resGsym2}).
\label{Asymptoticssym}
\end{theorem}

\begin{proof}
Applying equation~(\ref{Ldef}) to $s_2 = 0$ gives the Laplace transform of $U_1$ as 
\begin{equation}
F(s_1,0) = 1 - \varrho + F_0(s_1,0) + G(s_1) + F_0(0,s_1) + G(0)
\label{LaplU1}
\end{equation}
with 
$$
\left \{
\begin{array}{ll}
F_0(s_1,0) & = \displaystyle \frac{J(0)-K_2(s_1,0)G(0)}{K_1(s_1,0)}  +  \displaystyle \frac{H(s_1,0)}{K_1(s_1,0)}, \\  \\
F_0(0,s_1) & = \displaystyle \frac{J_1(s_1)-K_1(s_1,0)G(s_1)}{K_2(s_1,0)} - \displaystyle \frac{H(0,s_1)}{K_2(s_1,0)}, \\  \\
H(s_1,0) & = H(0,s_1) = \displaystyle \frac{-\lambda s_1}{2(\mu+s_1)}M\left( \frac{s_1}{2} \right),
\end{array} \right.
$$
by using \eqref{Fonctcomp} and (\ref{Hsym}). We now follow the results of Theorem \ref{Gsing} on the smallest singularity of $G$ in order to derive the smallest singularity of transform $s_1 \mapsto F(s_1,0)$ expressed above. 

$\bullet$ Assume first $\varrho > 1/2$. By  Proposition \ref{extendM}, function $s_1 \mapsto H(s_1,0)$ is analytic for $\Re(s_1) > 2(\sigma_0-\mu/2)/2 = \sigma_0 - \mu/2$. It then follows from (\ref{LaplU1}) that the singularity with smallest module of $F(s_1,0)$ is at $s_1 = \sigma_0$ with leading term
\begin{equation}
F(s_1,0) \sim - \frac{K_1(s_1,0)}{K_2(s_1,0)}G(s_1) + G(s_1) = \frac{2(s_1+\mu)}{\lambda} G(s_1)
\label{leadterm0}
\end{equation}  
since $K_1(s_1,0)/K_2(s_1,0) = -2(s_1+\mu-\lambda/2)/\lambda$ and the root $\lambda/2-\mu$ of $K_1(s_1,0)$ is a removable singularity since $F_1(s_1,0)$ has to be analytic for $\Re(s_1)>\tilde{s}$. By estimate (\ref{resGsym1}) for $G(s_1)$ near $s_1 = \sigma_0$, (\ref{leadterm0}) yields $F(s_1,0) \sim 2r_0/(s_1 - \sigma_0)$ as $s_1 \rightarrow \sigma_0$; smallest singularity $s_1 = \sigma_0$ is thus a simple pole for Laplace transform $s_1 \mapsto F(s_1,0)$. Applying then Theorem \ref{Tauberian} with $\kappa_0 = 2r_0$ and $\nu_0 = -1$, we derive that $\mathbb{P}(U_1 > u) \sim -  2r_0 e^{\sigma_0 u}/\sigma_0$ for large $u$ with prefactor
$$
-2 \frac{r_0}{\sigma_0} = -2 (1-\varrho)(2\varrho - 1)\frac{\mu}{-4\mu(1-\varrho)} = \varrho - \frac{1}{2}
$$
as claimed.

$\bullet$ Assume now that $\varrho < 1/2$. By formula (\ref{LaplU1}) and Proposition \ref{extendM}, function $s_1 \mapsto H(s_1,0)$ is analytic for $\Re(s_1) > 2\eta_1$. It then follows from (\ref{LaplU1}) that the singularity with smallest module of $F(s_1,0)$ is at $s_1 = \zeta^+$ with leading term again specified by (\ref{leadterm0}) so that
\begin{equation}
F(s_1,0) - F(\zeta^+,0) \sim \frac{2(\zeta^+ + \mu)}{\lambda} \left [ G(s_1) - G(\zeta^+) \right ]
\label{smallestpoleGter}
\end{equation}
near $s_1 = \zeta^+$. By estimate (\ref{resGsym2}), (\ref{smallestpoleGter}) yields $F(s_1,0) - F(\zeta^+,0) \sim r_1(s_1-\zeta^+)^{1/2}$ as $s_1 \rightarrow \zeta^+$ where
$$
r_1 = \frac{2(\zeta^+ + \mu)}{\lambda}r^+;
$$
smallest singularity $s_1 = \zeta^+$ is thus an algebraic singularity for Laplace transform $s_1 \mapsto F(s_1,0)$, with order $1/2$. Applying Theorem \ref{Tauberian} with $\kappa_0 = r_1$, $\nu_0 = 1/2$ and $\Gamma(-1/2) = -2 \sqrt{\pi}$, we derive that $\mathbb{P}(U_1 > u) \sim \kappa e^{\zeta^ +u}/u^{3/2}$ for large $u$ with prefactor $\kappa = -r_1/\zeta^+ \Gamma(-1/2)$.

$\bullet$ Finally, assume that $\varrho = 1/2$; the polar singularity $\sigma_0 = -\mu(1-\varrho) = -\mu/2$ and the algebraic singularity $\zeta^+ = -\mu/2$ for $G$ coincide in this case. Recall from Proposition \ref{extendM}.b that function $z \mapsto M \circ h(z)$ is analytic for $\Re(z) > \eta_1$ whenever $\varrho \leq 1/2$; $\eta_1$ is the only real zero $\neq -\mu$ of discriminant $\Delta(z) = (z+\mu)\delta(z)$ and expression (\ref{smalldelta}) of $\delta(z)$ gives $\delta(-\mu/2) = \mu^5/4 > 0$, hence $\eta_1 < -\mu/2$; $M \circ h$ is therefore analytic at $z = -\mu/2$. Near $s_1 = -\mu/2$, formula (\ref{xi+-}) easily gives
$$
\xi^+(s_1) = -\frac{\mu}{2} + \sqrt{\frac{\mu}{2}}\left(s_1+\frac{\mu}{2}\right)^{1/2} (1 + o(1));
$$
expression (\ref{Gsymbis}) for $G(s_1)$ and the discussion above then imply that
$$
G(s_1) = \frac{1}{4}\sqrt{\frac{\mu}{2}} \left (s_1+\frac{\mu}{2} \right )^{-1/2} (1 + o(1))
$$
in the neighborhood of $s_1 = -\mu/2$. The leading term (\ref{leadterm0}) for $F(s_1,0)$ is consequently given by
$$
F(s_1,0) \sim \frac{1}{2}\sqrt{\frac{\mu}{2}} \left (s_1+\frac{\mu}{2} \right )^{-1/2};
$$ 
smallest singularity $s_1 = -\mu/2$ is thus an algebraic singularity for Laplace transform $s_1 \mapsto F(s_1,0)$, with order $-1/2$. Applying then Theorem \ref{Tauberian} with $\kappa_0 = (\mu/2)^{1/2}/2$, $\nu_0 = -1/2$ and $\Gamma(1/2) = \sqrt{\pi}$, we derive that $\mathbb{P}(U_1 > u) \sim \kappa e^{-\mu +u/2}/u^{1/2}$ for large $u$ with prefactor $\kappa = 2\kappa_0/\mu \Gamma(1/2)$.
\end{proof}

For any given load $\varrho \in \; ]0,1[$, Theorem \ref{Asymptoticssym} consequently provides the same exponential trend as that of upper bound (\ref{asymptHoL}) for HoL; as a matter of fact, a large value of $U_1$ entails that queue $\sharp 1$ behaves as if queue $\sharp 2$, with smaller workload, had a HoL priority.


\section{Conclusion}
\label{CL}


The stationary analysis of two coupled queues addressed by a unique server running the SQF discipline has been generally considered for Poisson arrival processes and general service time distributions; required functional equations for the derivation of the stationary distribution for the coupled workload process have been derived. Specializing the resolution of such equations to both exponentially distributed service times and the so-called ``symmetric case'', all quantities of interest have been obtained by solving a single functional equation. 

The solution $M$ for that equation has been given, in particular, as a series expansion involving all consecutive iterates of an algebraic function $h$ related to a branch of some cubic equation $R(w,z) = 0$. It must be noted that the curve represented by that cubic equation in the $(O,w,z)$ plane is singular; in fact, whereas ``most'' cubic curves are regular (i.e., without multiple points), it can be easily checked that cubic $R = 0$ has a double point at infinity. In equivalent geometric terms, cubic $R = 0$ can be identified with a sphere when seen as a surface in $\mathbb{C} \times \mathbb{C}$, whereas most cubic curves are identified with a torus. This fact can be considered as an essential underlying feature characterizing the complexity of the present problem; such geometric statements will be enlightened  for solving the general asymmetric case in \cite{assymSQF}.

An extended analyticity domain for solution $M$ has been determined as the half-plane $V_M$, thus enabling to determine the singularity of Laplace transform $G$ with smallest module. It could be also of interest to compare such extended domain $V_M$ to the maximal convergence domain of series expansion (\ref{Mseries}) (recall the convergence of that series has been stated in Theorem \ref{resolsym} for real $z > 0$ only); in fact, the analyticity domain $V_M$ may not coincide with the validity domain for such a series representation. The discrete holomorphic dynamical system defined by the iterates $z \mapsto h^{(k)}(z)$, $k \geq 1$, definitely plays a central role for such a comparison. 

As an alternative approach to that of Section \ref{SQFqueue}, function $M$ may also be derived through a Riemann-Hilbert boundary value problem; hints for such an approach can be summarized as follows. We successively note that
\begin{itemize}
\item there exists $s_0 \in \; ]\zeta^-,\zeta^+[$ such that for $s > s_0 $, $\Re((\xi^\pm(s)+s)/2)$ belongs to the analyticity domain $V_M$ determined by Proposition~\ref{extendM};
\item denoting by $\mathcal{L}$ the image by functions $X^\pm:s \mapsto (\xi^\pm(s)+s)/2$ of the open interval $]s_0,\zeta^+[$, we note that $M(\overline{z})=\overline{M(z)}$ for $z \in \mathcal{L}$ with $z=X^+(s)$.  Equations~\eqref{MSymequ} then enable us to deduce the condition
\begin{equation}
\forall \; z \in \mathcal{L}, \; \; \Re\left(i\frac{\mu}{\mu+2z-s}M(z)\right)=\Im\left(\frac{(1-\rho)s}{2(s-z) } \right).
\label{BV0}
\end{equation} 
\end{itemize}
The above Riemann-Hilbert problem for function $M$ is, however, valid on open path $\mathcal{L}$ only and not on the whole closed contour $\partial \mathbf{D}$, defined as  the image by functions $X^\pm$ of closed segment $[\zeta^-,\zeta^+]$. The well-posed problem, nevertheless, formulates as follows.
\begin{problem}
\label{RHPglob}
Determine a function $\Phi$ which is analytic in $\C \setminus \mathbf{D}$, where $\mathbf{D}$ is the domain delineated by the closed contour $\partial \mathbf{D}$, tends to 0 at infinity and such that boundary condition (\ref{BV0}) holds on $\partial \mathbf{D}$ (and not only on $\mathcal{L}$).
\end{problem}
\noindent
If the solution $\Phi$ to \textbf{Problem 3} can be shown to exist and to be analytic on $\mathbf{D}$, then functions $M$ and $\Phi$ coincide. Proving the latter statement and deriving an alternative representation of solution $M$ (namely, as a path integral on closed contour $\partial \mathbf{D}$) is an object of further study.

On the application side, the performance of the SQF discipline has been  characterized, both in terms of empty queue probability and distribution tail at infinity. The results show that SQF compares quite favorably with respect to the ``optimal'' priority discipline, namely HoL. Such performance properties will be generalized to the asymmetric case where flow patterns are allowed to be heterogeneous.


\appendix
\section{Proof for Assertion b) of Proposition~\ref{resol}}
\label{A2}


Before proving equations (\ref{Fonctcomp}), we state preliminary  expressions of $G_1$ and $G_2$.
\begin{lemma}
Given
\begin{equation}
\label{defE21E12}
E_{21}(s_1) = \int_0^{+\infty} e^{-s_1u_1} \varphi_2(u_1,0) \mathrm{d}u_1,  \; \; \; E_{12}(s_2) = \int_0^{+\infty} e^{-s_2u_2} \varphi_1(0,u_2) \mathrm{d}u_2,
\end{equation}
univariate transforms $G_1$ and $G_2$ satisfy 
\begin{equation}
\left\{ 
\begin{array}{ll}
G_1(s_1) = \displaystyle \frac{(1-\varrho)(\lambda - \lambda_1b_1(s_1)) - E_{21}(s_1) - \psi_2(0)}{s_1 - \lambda + \lambda_1b_1(s_1)}, \; \; \; \Re(s_1) > 0,
\\ \\
G_2(s_2) = \displaystyle \frac{(1-\varrho)(\lambda - \lambda_2b_2(s_2)) - E_{12}(s_2) - \psi_1(0)}{s_2 - \lambda + \lambda_2b_2(s_2)}, \; \; \; \Re(s_2) > 0.
\label{G1G2comp}
\end{array} \right.
\end{equation}
\label{G1G2prov}
\end{lemma}

\begin{proof}
As transforms of regular densities, we have  $b_2(s_2) \rightarrow 0$, $F_1(s_1,s_2) \rightarrow 0$ when $s_2 \rightarrow +\infty$ for fixed $s_1$ with  $\Re(s_1) > 0$. Besides, we have  $s_2F_2(s_1,s_2) \rightarrow E_{21}(s_1)$, $s_2G_2(s_2) \rightarrow \psi_2(0)$ when $s_2 \rightarrow +\infty$ with fixed $s_1, \; \Re(s_1) > 0$, where $E_{21}$ is the Laplace transform of the restriction of density $\varphi_2$ on the boundary $\delta_1$ and $\psi_2(0)$ is the value at $u_2 = 0_+$ of density $\psi_2$ on boundary $\delta_2$; as a consequence,
$$
\lim_{s_2 \rightarrow +\infty} s_2H_2(s_1,s_2) = \lim_{s_2 \rightarrow +\infty} s_2(F_2(s_1,s_2) + G_2(s_2)) = E_{21}(s_1) + \psi_2(0)
$$
for fixed $s_1, \; \Re(s_1) > 0$. Now, letting $s_2$ tend to $+\infty$ in each side of  (\ref{Fonct}), the above limit results entail $(s_1 - K(s_1,\infty))G_1(s_1) + E_{21}(s_1) + \psi_2(0) = (1-\varrho)K(s_1,\infty)$ with $K(s_1,\infty) = \lambda - \lambda_1b_1(s_1)$, which provides identity (\ref{G1G2comp}) for $G_1(s_1)$. Identity (\ref{G1G2comp}) for $G_2(s_2)$ is symmetrically deduced by letting $s_1$ tend to $+\infty$ in (\ref{Fonctcomp}) with fixed $s_2, \; \Re(s_2) > 0$.
\end{proof}

We now address the derivation of equations (\ref{Fonctcomp}). Recall that subsets $\Gamma_1$, $\delta_1$, etc. of state space $\mathcal{U}$ are defined in  (\ref{cones12})-(\ref{axes12}). Given $\varepsilon > 0$, define the function $Y_\varepsilon$ by $Y_\varepsilon(v) = \exp(-\varepsilon/v)\ind_{\{v > 0\}}$; $Y_\varepsilon$ is twice continuously differentiable over $\mathbb{R}$,  $\lim_{\varepsilon \downarrow 0} Y_\varepsilon(v) = \ind_{\{v > 0\}}$ for each $v \in \mathbb{R}$ and $\lim_{\varepsilon \downarrow 0} Y_\varepsilon' = \delta_0$ (the Dirac mass at $v = 0$) for the weak convergence of distributions. For given $\Re(s_1) > 0$, $\Re(s_2) > 0$, let then be the test function $\theta_\varepsilon(\mathbf{u}) = e^{-\mathbf{s} \cdot \mathbf{u}}\chi_\varepsilon(\mathbf{u})$, $\mathbf{u} \in \mathcal{U}$, with 
\begin{equation}
\chi_\varepsilon(\mathbf{u}) = Y_\varepsilon(u_1)Y_\varepsilon(u_2-u_1).
\label{Iepsilon}
\end{equation}
Function $\theta_\varepsilon$ belongs to $\mathcal{C}_b^2(\mathcal{U})$ and is 0 on the outside of $\Gamma_1$; moreover, we have $\lim_{\varepsilon \downarrow 0} \chi_\varepsilon(\mathbf{u}) = \ind_{\{\mathbf{u} \in \Gamma_1\}}$ so that $\lim_{\varepsilon \downarrow 0} \theta_\varepsilon = \theta$ pointwise in $\mathcal{U}$, with limit function $\theta$ defined by $\theta(\mathbf{u}) = e^{-\mathbf{s} \cdot \mathbf{u}}\ind_{\{\mathbf{u} \in \Gamma_1\}}$, $\mathbf{u} \in \mathcal{U}$.

By direct differentiation, we further calculate
$$
\frac{\partial\theta_\varepsilon}{\partial u_1}(\mathbf{u}) = -s_1\theta_\varepsilon(\mathbf{u}) + e^{-\mathbf{s}\cdot\mathbf{u}} \frac{\partial \chi_\varepsilon}{\partial u_1}(\mathbf{u}),
\; \; \; \; 
\displaystyle \frac{\partial\theta_\varepsilon}{\partial u_2}(\mathbf{u}) = -s_2\theta_\varepsilon(\mathbf{u}) + e^{-\mathbf{s}\cdot\mathbf{u}} \frac{\partial \chi_\varepsilon}{\partial u_2}(\mathbf{u})
$$
for $\mathbf{u} \in \mathcal{U}$, with
$$
\frac{\partial \chi_\varepsilon}{\partial u_1}(\mathbf{u}) = Y_\varepsilon'(u_1)Y_\varepsilon(u_2-u_1) - Y_\varepsilon(u_1)Y_\varepsilon '(u_2-u_1)
$$
after (\ref{Iepsilon}); note that derivative $\partial \chi_\varepsilon / \partial u_1$ tends to $\delta_0(u_1) - \delta_0(u_2 - u_1)$ for the weak convergence of distributions as $\varepsilon \downarrow 0$.
\\
\indent
Let us now calculate the limit $\mathcal{M} = \lim_{\varepsilon \downarrow 0} \mathcal{M}(\varepsilon)$ with $\mathcal{M}(\varepsilon)$ introduced in (\ref{mepsilon}); to this end, we address successive terms of $\smallint_{\mathcal{U}}\mathcal{A}\theta_\varepsilon \mathrm{d}\Phi$ according to definition (\ref{gen1}). Integrating first $\partial \theta_\varepsilon/\partial u_1$ over $\Gamma_1 \cup \delta_1$ against  $\mathrm{d}\Phi$ reduces to
$$
- \int_{\Gamma_1 \cup \delta_1} \frac{\partial\theta_\varepsilon}{\partial u_1} \mathrm{d}\Phi = - \int_{\Gamma_1} \frac{\partial\theta_\varepsilon}{\partial u_1}(\mathbf{u})\varphi_1(\mathbf{u})\mathrm{d}\mathbf{u} 
= \int_{\Gamma_1} \left[ s_1 \theta_\varepsilon(\mathbf{u}) - e^{-\mathbf{s} \cdot \mathbf{u}} \frac{\partial \chi_\varepsilon}{\partial u_1}(\mathbf{u}) \right ] \varphi_1(\mathbf{u})\mathrm{d}\mathbf{u} 
$$
since $\partial \theta_\varepsilon/\partial u_1$ vanishes on the outside of $\Gamma_1$; on  account of the above mentioned weak convergence properties, we then obtain
\begin{equation}
\lim_{\varepsilon \downarrow 0} - \int_{\Gamma_1 \cup \delta_1} \frac{\partial\theta_\varepsilon}{\partial u_1}(\mathbf{u})\mathrm{d}\Phi(\mathbf{u}) = s_1F_1(s_1,s_2) - E_{12}(s_2) + Z_1(s_1+s_2)
\label{lim1}
\end{equation}
with $E_{12}(s_2)$ defined as in Lemma \ref{G1G2prov} and where $Z_1(s) = \smallint_{u > 0} \; e^{-su} \varphi_1(u,u) \mathrm{d}u$ defines the Laplace transform of density $\varphi_1$ restricted to the positive diagonal $\delta$ (function $Z_1$ is determined below). Besides, the integral of $\partial \theta_\varepsilon/\partial u_2$ over $\Gamma_2 \cup \delta_2$ equals 0 as this function vanishes on the outside of $\Gamma_1$.

Further, we have  $\lim_{\varepsilon \downarrow 0} \theta_\varepsilon(\mathbf{u} + \mathcal{T}_1\mathbf{e}_1) = e^{-\mathbf{s} \cdot \mathbf{u}} e^{-s_1\mathcal{T}_1} \ind_{\{\mathcal{T}_1 < u_2 - u_1\}}$ for given $\mathbf{u} \in \mathcal{U}$ and $\mathcal{T}_1 > 0$, therefore $\lim_{\varepsilon \downarrow 0} \mathbb{E}\theta_\varepsilon(\mathbf{u} + \mathcal{T}_1\mathbf{e}_1) = e^{-\mathbf{s} \cdot \mathbf{u}} \mathbb{E}(e^{-s_1\mathcal{T}_1} \ind_{\{\mathcal{T}_1 < u_2 - u_1\}})$ by the Dominated Convergence theorem; hence
\begin{multline}
\label{lim3}
 \lim_{\varepsilon \downarrow 0} \; \mathbb{E}\theta_\varepsilon(\mathbf{U} + \mathcal{T}_1\mathbf{e}_1) = \int_{\mathcal{U}} e^{-\mathbf{s} \cdot \mathbf{u}} \mathbb{E} \left ( e^{-s_1\mathcal{T}_1} \ind_{\{\mathcal{T}_1 < u_2 - u_1\}} \right ) \mathrm{d}\Phi(\mathbf{u}) =
\\
 \int_{\Gamma_1} e^{-\mathbf{s} \cdot \mathbf{u}} \mathbb{E} \left ( e^{-s_1\mathcal{T}_1} \ind_{\{\mathcal{T}_1 < u_2 - u_1\}} \right )\varphi_1(\mathbf{u})\mathrm{d}\mathbf{u} 
+ \int_0^{+\infty} e^{-s_2u_2}\mathbb{E} \left ( e^{-s_1\mathcal{T}_1} \ind_{\{\mathcal{T}_1 < u_2\}} \right )\psi_2(u_2)\mathrm{d}u_2
\end{multline}
where random variable $\mathbf{U} = (U_1,U_2)$ has distribution $\Phi$. For given $\mathbf{u} \in \mathcal{U}$, we similarly have $\lim_{\varepsilon \downarrow 0} \mathbb{E}\theta_\varepsilon(\mathbf{u} + \mathcal{T}_2\mathbf{e}_2) = e^{-\mathbf{s} \cdot \mathbf{u}} \mathbb{E}(e^{-s_2\mathcal{T}_2} \ind_{\{\mathcal{T}_2 > u_1 - u_2, u_1 > 0\}})$
and therefore
\begin{align}
\label{lim4}
\lim_{\varepsilon \downarrow 0} \; \mathbb{E}\theta_\varepsilon(\mathbf{U} + \mathcal{T}_2\mathbf{e}_2) & = \int_{\mathcal{U}} e^{-\mathbf{s} \cdot \mathbf{u}} \mathbb{E} \left ( e^{-s_2\mathcal{T}_2} \ind_{\{\mathcal{T}_2 > u_1 - u_2, u_1 > 0\}} \right ) \mathrm{d}\Phi(\mathbf{u}) = 
\\ 
b_2(s_2)F_1(s_1,s_2)  & + 
\int_{\Gamma_2} e^{-\mathbf{s} \cdot \mathbf{u}} \mathbb{E} \left ( e^{-s_2\mathcal{T}_2} \ind_{\{\mathcal{T}_2 > u_1 - u_2\}} \right ) \varphi_2(\mathbf{u}) \mathrm{d}\mathbf{u} 
\nonumber \\ 
& + 
\int_0^{+\infty} e^{-s_1u_1}\mathbb{E} \left ( e^{-s_2\mathcal{T}_2} \ind_{\{\mathcal{T}_2 > u_1\}} \right )\psi_1(u_1)\mathrm{d}u_1.
\nonumber
\end{align}
\noindent
Finally, noting that $\lim_{\varepsilon \downarrow 0} \smallint_{\mathcal{U}} \theta_\varepsilon(\mathbf{u}) \mathrm{d}\Phi(\mathbf{u}) = F_1(s_1,s_2)$ and adding limit terms (\ref{lim1}), (\ref{lim3}), (\ref{lim4}) according to (\ref{gen1}) gives limit $\mathcal{M} = \lim_{\varepsilon \downarrow 0} \mathcal{M}(\varepsilon) = 0$ the final expression
\begin{align*}
& \; s_1 F_1(s_1,s_2) - E_{12}(s_2) + Z_1(s_1 + s_2) 
\nonumber \\
& + \lambda_1 \left [ b_1(s_1)F_1(s_1,s_2) - \int_{\Gamma_1} e^{-\mathbf{s} \cdot \mathbf{u}} \mathbb{E} \left ( e^{-s_1\mathcal{T}_1} \ind_{\{\mathcal{T}_1 > u_2 - u_1\}} \right )\varphi_1(\mathbf{u})\mathrm{d}\mathbf{u} \right ]
 \\
& + \lambda_1 \left [ b_1(s_1)G_2(s_2) - \int_0^{+\infty} e^{-s_2u_2}\mathbb{E} \left ( e^{-s_1\mathcal{T}_1} \ind_{\{\mathcal{T}_1 > u_2\}} \right )\psi_2(u_2)\mathrm{d}u_2 \right ]  \\
&+ \lambda_2 b_2(s_2) F_1(s_1,s_2)  + \lambda_2 \; \bigg [ \int_{\Gamma_2} e^{-\mathbf{s} \cdot \mathbf{u}} \mathbb{E} \left ( e^{-s_2\mathcal{T}_2} \ind_{\{\mathcal{T}_2 > u_1 - u_2\}} \right ) \varphi_2(\mathbf{u}) \mathrm{d}\mathbf{u} 
 \\
&+  \int_0^{+\infty} e^{-s_1u_1}\mathbb{E} \left ( e^{-s_2\mathcal{T}_2} \ind_{\{\mathcal{T}_2 > u_1\}} \right  )\psi_1(u_1)\mathrm{d}u_1 \bigg ] - \lambda F_1(s_1,s_2) = 0.
\end{align*}

Defining $H(s_1,s_2)$ as in (\ref{H}) to gather all remaining integrals, the latter identity reads 
\begin{equation}
K_1(s_1,s_2)F_1(s_1,s_2) + \lambda_1b_1(s_1) G_2(s_2) = E_{12}(s_2) - Z_1(s_1 + s_2) + H(s_1,s_2)
\label{Fonctcomp1}
\end{equation}
with $K_1(s_1,s_2) = s_1-K(s_1,s_2)$, $E_{12}(s_2)$ being defined by (\ref{defE21E12}) and where $Z_1$ defines the Laplace transform of density $\varphi_1$ restricted to the diagonal. Changing index 1 into 2, and noting that $H(s_1,s_2)$ changes into $-H(s_1,s_2)$, symmetrically yields second equation
\begin{equation}
K_2(s_1,s_2)F_2(s_1,s_2) + \lambda_2b_2(s_2) G_1(s_1) = E_{21}(s_1) - Z_2(s_1 + s_2) - H(s_1,s_2)
\label{Fonctcomp2}
\end{equation}
with $K_2(s_1,s_2) = s_2 - K(s_1,s_2)$, $E_{21}(s_1)$ defined by \eqref{defE21E12} and where $Z_2$ defines the Laplace transform of density $\varphi_2$ restricted to the diagonal. To conclude the proof, we prove the following technical lemma.

\begin{lemma}
Functions $Z_1$ and $Z_2$ are identically zero.
\label{W1W2}
\end{lemma}

\begin{proof}
Adding equations (\ref{Fonctcomp1}) and (\ref{Fonctcomp2}) (and omitting arguments for the sake of simplicity) yields $K_1F_1 + K_2F_2 = -\lambda_1b_1G_2 - \lambda_2b_2G_1 + E_{12} - Z_1 + E_{21} - Z_2$. On the other hand, equation (\ref{Fonct}) gives $K_1F_1 + K_2F_2 = (1-\varrho)K - K_1G_1 - K_2G_2$; equating right hand sides of the latter equations then provides the identity 
\begin{align}
& (1-\varrho)K(s_1,s_2) - (s_1-\lambda+\lambda_1b_1(s_1))G_1(s_1) - (s_2-\lambda+\lambda_2b_2(s_2))G_2(s_2) = 
\nonumber \\
& E_{12}(s_2) + E_{21}(s_1) - Z_1(s_1+s_2) - Z_2(s_1+s_2).
\nonumber
\end{align}
Using expressions (\ref{G1G2comp}) for $G_1(s_1)$ and $G_2(s_2)$, the latter identity simply reduces to $Z_1(s_1+s_2) + Z_2(s_1 + s_2) = \lambda(1-\varrho) - \psi_1(0) - \psi_2(0)$, showing that function $Z_1 + Z_2$ is constant. As both $Z_1$ and $Z_2$ vanish at $+\infty$, this constant is 0 and since these functions are non negative by definition, this entails that $Z_1 = Z_2 = 0$.
\end{proof}

After using equations (\ref{G1G2comp}) to  express $E_{12}(s_2)$ and $E_{21}(s_1)$ in terms of $G_2(s_2)$ and $G_1(s_1)$, respectively, Lemma \ref{W1W2} finally enables us to reduce (\ref{Fonctcomp1}) and (\ref{Fonctcomp2}) to equations (\ref{Fonctcomp}). This concludes the proof of Proposition \ref{resol}.


\section{Proof of Proposition \ref{Hexp}}
\label{A3}


For an exponentially distributed service time $\mathcal{T}_1$ with parameter $\mu_1$, the factor of $\lambda_1$ in definition (\ref{H}) of $H(s_1,s_2)$ reads 
\begin{multline*}
\E \left [ e^{-s_1U_1-s_2U_2}\ind_{\{0 \leq U_1 < U_2\}}e^{-s_1\mathcal{T}_1}\ind_{\{\mathcal{T}_1 > U_2-U_1\}} \right ] = \\
\int_0^{+\infty} \int_0^{+\infty} \left [ \int_{u_2-u_1}^{+\infty} e^{-s_1x_1} \mu_1 e^{-\mu_1x_1} \mathrm{d}x_1 \right ]e^{-s_1u_1-s_2u_2}\ind_{\{0 \leq u_1 < u_2\}}\mathrm{d}\Phi(u_1,u_2),
\end{multline*}
for $\Re(s_1)\geq $ and $\Re(s_2)\geq 0$. By definition (\ref{dPhi}), the latter term is equal to 
\begin{multline*}
 \int_0^{+\infty} \int_0^{+\infty} \frac{\mu_1}{\mu_1+s_1}e^{\mu_1u_1-(s_1+s_2+\mu_1)u_2}\ind_{\{0 \leq u_1 < u_2\}} \Big [ \psi_2(u_2)\ind_{\{u_2 > 0\}}\mathrm{d}u_2 \otimes \delta_0(u_1) \; +
 \\
 \varphi_1(u_1,u_2)  \mathrm{d}u_1\mathrm{d}u_2 \Big ]  \\ 
= \frac{\mu_1}{\mu_1+s_1} \big [ G_2(s_1+s_2+\mu_1) + F_1(-\mu_1,s_1+s_2+\mu_1) \big ]
\end{multline*}
where, by Corollary \ref{extensions}, each term inside brackets is analytically defined for $(s_1,s_2)$ such that $\Re(s_1+s_2 + \mu_1) > \widetilde{s}_2$ and $\Re(s_1+s_2+\mu_1) > \max(\widetilde{s}_2,\widetilde{s}_2+\mu_1) = \widetilde{s}_2 + \mu_1$, respectively, that is at least for $\Re(s_1 + s_2) > \widetilde{s}_2$. Similarly, for an exponentially distributed service time $\mathcal{T}_2$ with parameter $\mu_2$, the factor of $\lambda_2$ in definition (\ref{H}) of $H(s_1,s_2)$ reads  
\begin{multline*}
- \E \left [ e^{-s_1U_1-s_2U_2}\ind_{\{0 \leq U_2 < U_1\}}e^{-s_2\mathcal{T}_2}\ind_{\{\mathcal{T}_2 > U_1-U_2\}} \right ] =  \\ - \frac{\mu_2}{\mu_2+s_2}\left [ G_1(s_1+s_2+\mu_2) + F_2(s_1+s_2+\mu_2,-\mu_2)\right ]
\end{multline*}
where, by Corollary \ref{extensions}, each term inside brackets is analytically defined for $(s_1,s_2)$ such that $\Re(s_1+s_2 + \mu_2) > \widetilde{s}_1$ and $\Re(s_1+s_2+\mu_2) > \max(\widetilde{s}_1,\widetilde{s}_1+\mu_2) = \widetilde{s}_1 + \mu_2$, respectively, hence for $\Re(s_1 + s_2) > \widetilde{s}_1$. Adding up the two above expressions, we obtain claimed expressions. 


\section{Proof of Proposition \ref{extendM}}
\label{A4}


With $s = z - \alpha(z)$ and $\xi^-(s) = 2z - s = z + \alpha(z)$, second equation (\ref{MSymequ}) reads
\begin{align}
M(z) = \frac{z + \alpha(z)+\mu}{\mu} \Big[ & \frac{2(z-\alpha(z)+\mu)}{\lambda}G(z-\alpha(z)) 
\nonumber \\
& + \; (1-\varrho)\frac{(z+\alpha(z)+\mu)}{2\alpha(z)} \Big].
\label{MG}
\end{align}
We successively make the following points:
\begin{itemize}
\item By Lemma \ref{discriminant}, function $z \mapsto \alpha(z)$ is analytic on the cut plane $\mathbb{C} \setminus [\eta_2,\eta_1]$, where ramification points $\eta_2$, $\eta_1$ are determined as the real negative roots of discriminant $\Delta(z)$. As $\eta_2 =-\mu < \eta_1 < 0$, function $z \mapsto \alpha(z)$ is, in particular, analytic in the half-plane $\{z \in \mathbb{C} \mid \; \Re(z) > \eta_1\}$;
\item By definition (\ref{cubiceq}), we may have  $\alpha(z) = 0$ only if $z(z+\mu)(z+\mu-\lambda)= 0$, that is, $z = 0$ or $z = -\mu$ or $z=\sigma_0=-\mu(1-\varrho)$; in the case $z=0$, we have
$$
\alpha(0) = \frac{\lambda-\sqrt{\lambda^2+4\mu^2}}{2}< \beta(0)=0 < \gamma(0) =  \frac{\lambda+\sqrt{\lambda^2+4\mu^2}}{2}
$$
and in the case $z=\sigma_0$,
$$
\alpha(\sigma_0) = \frac{-\mu-\sqrt{\mu^2+4\lambda^2}}{2}< \beta(\sigma_0)=0 < \gamma(\sigma_0) = \frac{-\mu+\sqrt{\mu^2+4\lambda^2}}{2};
$$
we conclude that we cannot have $\alpha(z)=0$ if $\Re(z)>\eta_1$;
\item By Corollary \ref{extensions}, transform $G$ is analytic on $\widetilde{\omega} = \{s \in \mathbb{C} \mid \; \Re(s) > \widetilde{s} \}$ where $\widetilde{s} = \sigma_0 = -\mu(1-\varrho)$ if $\varrho > 1/2$ and $\widetilde{s} = \zeta^+$ if $\varrho < 1/2$.
\end{itemize}
From expression (\ref{MG}) and the latter observations, we deduce that $M$ is analytic at any point $z$ with $\Re(z) > \eta_1$ and 
\begin{equation}
\Re(A(z)) > \widetilde{s}
\label{conditionAA}
\end{equation}
where $A(z) = z - \alpha(z)$. 
 
\begin{figure}[b]
\centering
\scalebox{.40}{\includegraphics{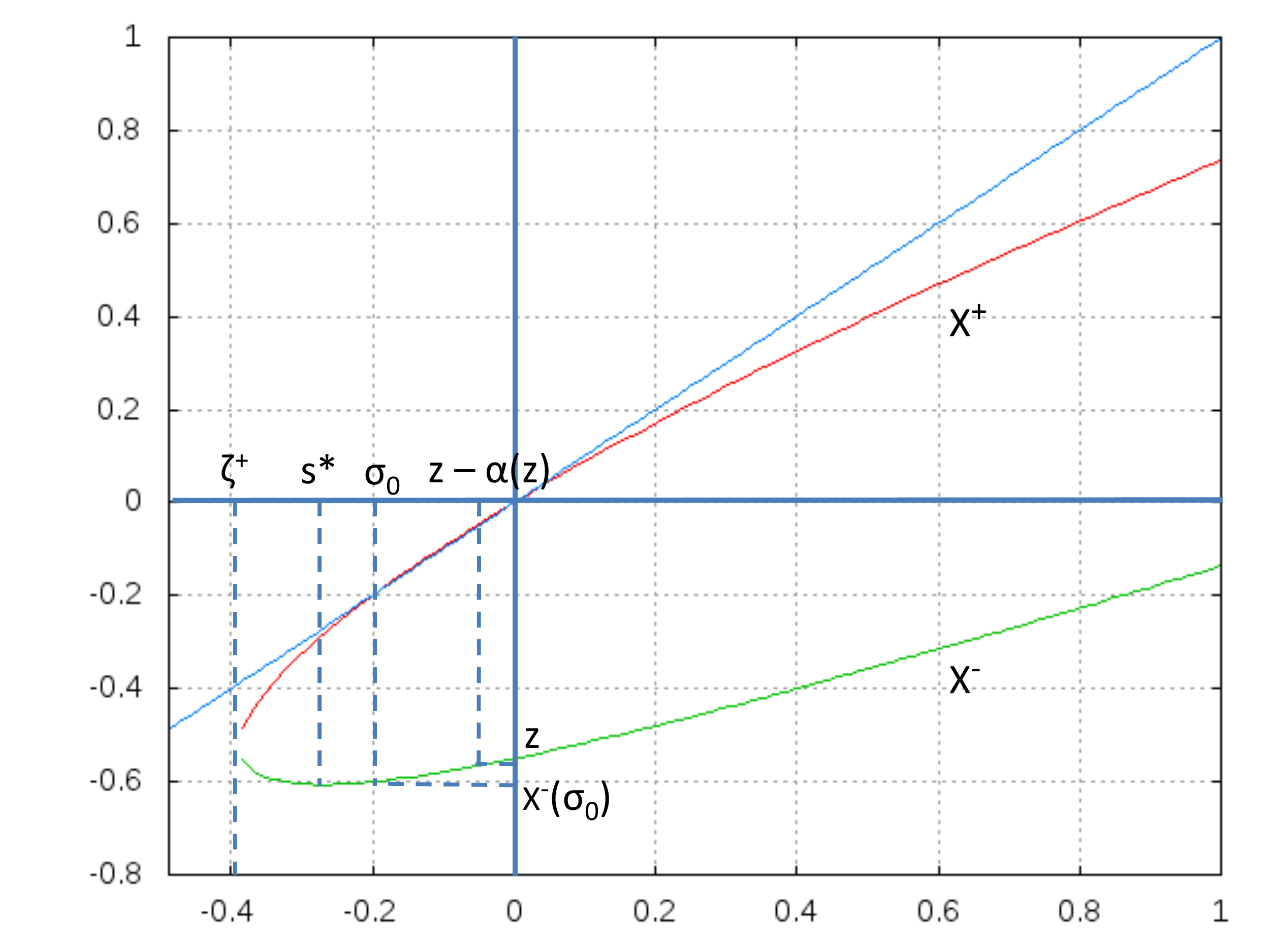}}
\caption{Case $\varrho>1/2$ ($\lambda=1.8$, $\mu=2$; $\varrho=0.9$).\label{courbexi1}}
\end{figure}

\textbf{a)} Assume first that $\varrho>1/2$. In the $(O,z,s)$ plane, the diagonal $z=s$ intersects the curve $z = X^+(s) = (s+\xi^+(s))/2$ at $s=\sigma_0$ (see Fig.~\ref{courbexi1}). Further, we easily verify that $A(z) = z-\alpha(z)>\sigma_0$ for 
$$
z > \frac{\sigma_0+\xi^-(\sigma_0)}{2} = \frac{1}{2} \left( \sigma_0-\frac{\mu}{2} \right)
$$
and condition (\ref{conditionAA}) is therefore fulfilled in this first case. We then conclude that function $M$ is analytic for $z > \frac{1}{2}(\sigma_0 - \frac{\mu}{2})$, and thus for $\Re(z) > \frac{1}{2}(\sigma_0 - \frac{\mu}{2})$ (recall by definition (\ref{Msym}) that $M$ is the sum of two non-negative Laplace transforms).

\textbf{b)} Assume now that $\varrho \leq 1/2$ (see Fig.~\ref{courbexi2}). We have shown above that we cannot have $\sigma_0=\sigma_0-\alpha(\sigma_0)$, which would otherwise imply $\alpha(\sigma_0)=0$. We thus necessarily have $\sigma_0 < s^*$, which entails that $A(z) = z-\alpha(z) > \sigma_0$ for $z>\eta_1$ and condition (\ref{conditionAA}) is therefore fulfilled in this second case. We then conclude that  function $M$ is analytic for $z>\eta_1$, hence for $\Re(z) > \eta_1$.

\begin{figure}[hbtp]
\centering
\scalebox{.40}{\includegraphics{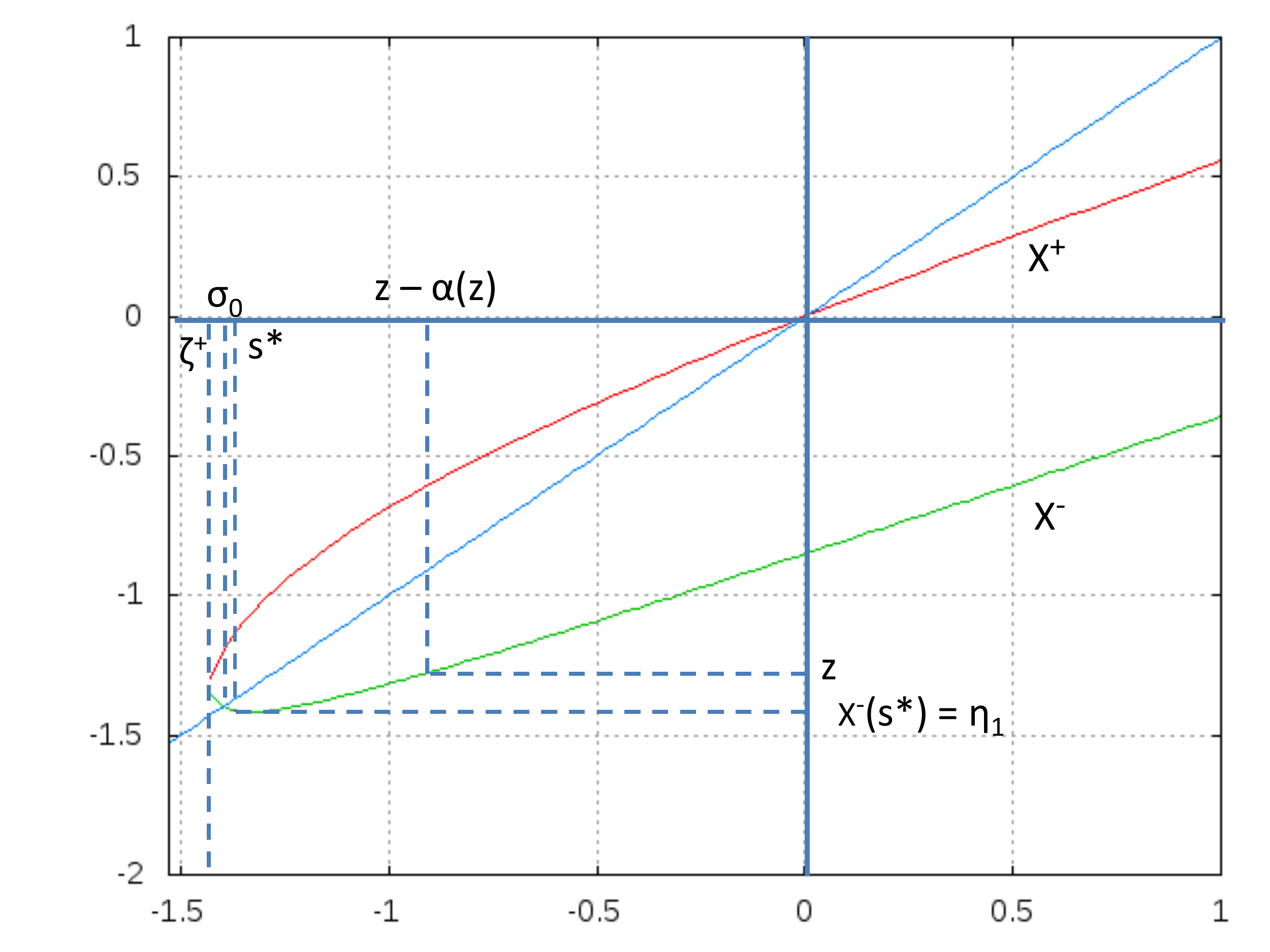}}
\caption{Case $\varrho<1/2$  ($\lambda=0.6$, $\mu=2$; $\varrho=0.3$).\label{courbexi2}}
\end{figure}


\end{document}